\pdfoutput=1
\documentclass[acmsmall]{acmart}
\AtBeginDocument{%
  }

\renewcommand\footnotetextcopyrightpermission[1]{} 
\usepackage{bigstrut}
\usepackage{amsfonts}
\usepackage{amsmath}
\usepackage{amsthm}
\usepackage{mathtools}
\usepackage{stmaryrd}
\usepackage{graphicx}
\usepackage[utf8]{inputenc}
\usepackage{xspace}
\usepackage[inline]{enumitem} 
\usepackage{booktabs} 

\usepackage[normalem]{ulem}

\usepackage{hyperref}
\usepackage{thmtools}
\usepackage{thm-restate}
\usepackage{cleveref}
\hypersetup{
    colorlinks=true,
    linkcolor=blue,
    filecolor=magenta,      
    urlcolor=cyan,
    }
\usepackage{tikz}
\usetikzlibrary{trees}
\usetikzlibrary{fit,positioning}

\usepackage{pgfplots}
\pgfplotsset{compat = newest}

\usepackage{color}

\newcommand\jonni[1]{\textcolor{red}{(Jonni) #1}}

\newtheorem{claim}{\textit{Claim}}

\theoremstyle{definition}

\newtheorem{remark}{Remark}[section]

\newcommand{\eqdef}{\mathrel{{\mathop:}}=}

\newcommand{\N}{\mathbb N}
\newcommand{\R}{\mathbb R}
\newcommand{\B}{\mathbb B}
\newcommand{\K}{\mathbb K}
\newcommand{\C}{\mathcal C}

\newcommand{\PTIME}{\ensuremath{\mathsf{P}}\xspace}
\newcommand{\DLOGTIME}{\ensuremath{\mathsf{DLOGTIME}}\xspace}
\newcommand{\NL}{\ensuremath{\mathsf{NL}}\xspace}
\newcommand{\Log}{\ensuremath{\mathsf{L}}\xspace}
\newcommand{\coNP}{\ensuremath{\mathsf{coNP}}\xspace}
\newcommand{\NP}{\ensuremath{\mathsf{NP}}\xspace}

\newcommand{\aco}{\ensuremath{\mathsf{AC}^0}\xspace}

\newcommand{\faco}{\ensuremath{\mathsf{FnAC}^0}\xspace}

\newcommand{\eval}[3]{{#1}({#2},{#3})}
\newcommand{\FOKeval}[2]{{#1}({{#2}})}

\newcommand{\qminusR}{\ensuremath{q \setminus R}}
\newcommand{\qyzminusR}{\ensuremath{\widetilde{q}}}

\newcommand{\Rep}{\mathrm{Rep}}
\newcommand{\dv}{\mathrm{A}}
\newcommand{\m}{\mathrm{m}}

\newcommand{\supp}{\overline{\mathrm{Supp}}}
\newcommand{\suppo}{\mathrm{Supp}}

\newcommand{\dfn}{\coloneqq}
\newcommand{\Var}{\mathsf{Var}}

\newcommand{\closureVars}{\ensuremath{(\var(\vec{y}) )^+_{\Sigma(q\setminus R)}}\xspace}

\newcommand{\enc}{\mathrm{enc}}

\newcommand{\epath}{\ensuremath{e_{path}(\DB)}\xspace}
\newcommand{\Logic}{\ensuremath{\mathcal{L}}}
\newcommand{\ICs}{\ensuremath{\Sigma}}
\newcommand{\DB}{\mathfrak{D}}
\newcommand{\RP}{\mathfrak{R}}
\newcommand{\RPS}{\mathfrak{S}}
\newcommand{\CA}{\ensuremath{\mathsf{CA}}\xspace}

\newcommand{\sjfCQ}{\ensuremath{\mathrm{sjfCQ}}\xspace}

\newcommand{\qpath}{\ensuremath{\mathrm{q}_\text{path}}\xspace}
\newcommand{\qcycle}{\ensuremath{\mathrm{q}_\text{cycle}}\xspace}
\newcommand{\qsink}{\ensuremath{\mathrm{q}_\text{sink}}\xspace}

\newcommand{\FO}{\ensuremath{\mathsf{FO}}\xspace}

\newcommand{\var}{\mathsf{var}}
\newcommand{\key}[1]{\mathsf{key}(#1)}

\definecolor{darkred}{rgb}{0.5, 0.0, 0.0}
\usepackage[backgroundcolor=orange!20, textcolor={darkred}, textsize=tiny]{todonotes}
\definecolor{green}{RGB}{0,120,0}
\definecolor{hlyellow}{RGB}{250, 250, 190}
\definecolor{ninaeditcolor}{RGB}{160,195,185}
\definecolor{almostwhite}{rgb}{0.9,0.9,0.9}
\definecolor{phokioneditcolor}{RGB}{190,190,185}
\definecolor{jonnieditcolor}{RGB}{190,160,195}
\definecolor{jefeditcolor}{RGB}{0,255,255}

\newcommand{\phokion}[1]{\todo[inline,caption={},color=phokioneditcolor, size=\footnotesize]{{\bf Ph:} #1}}

\newcommand{\cons}{\textsc{Cons}}
\newcommand{\CONS}{\textsc{Cons}}
\newcommand{\block}{\mathcal{B}}
\newcommand{\lrformula}[1]{\left({#1}\right)}
\newcommand{\agree}[2]{{\mathsf{PreCopy}}(#1,#2)}
\newcommand{\compose}[2]{\theta_{#2}^{#1}}

\newcommand{\closed}{closed}

\newcommand{\commentout}[1]{}

\newcommand{\suggestion}[2]{#2}

\makeatletter
\def\@copyrightspace{\relax}
\makeatother
\setcopyright{none} 

\begin{document}

\title{Rewriting Consistent Answers on Annotated Data}


\author{Phokion Kolaitis}
\affiliation{%
  \institution{University of California Santa Cruz \& IBM Research}
  \city{Santa Cruz}
  \state{California}
  \country{USA}
}
\email{kolaitis@soe.ucsc.edu}

\author{Nina Pardal}
\affiliation{%
  \institution{University of Huddersfield, University of Sheffield}
  \city{Huddersfield}
  \country{United Kingdom}}
\email{n.pardal@hud.ac.uk}

\author{Jonni Virtema}
\affiliation{%
  \institution{University of Sheffield}
  \city{Sheffield}
  \country{United Kingdom}}
\email{j.t.virtema@sheffield.ac.uk}

\author{Jef Wijsen}
\affiliation{%
  \institution{University of Mons}
  \city{Mons}
  \country{Belgium}}
\email{jef.wijsen@umons.ac.be}

\renewcommand{\shortauthors}{Kolaitis et al.}

\begin{abstract}
We embark on a study of the consistent answers of queries over databases annotated with values from a naturally ordered positive semiring. 
In this setting, the consistent answers of a query are defined as the minimum of the semiring values that the query takes over all repairs of an inconsistent database. 
The main focus is on self-join free conjunctive queries and key constraints, which is the most extensively studied case of consistent query answering over standard databases. We introduce a variant of first-order logic with a limited form of negation, define suitable semiring semantics, and then establish the main result of the paper: the consistent query answers of a self-join free conjunctive query under key constraints are rewritable in this logic if and only if the attack graph of the query contains no cycles. This result generalizes an analogous result of Koutris and Wijsen for ordinary databases, but also yields new results for a multitude of semirings, including the bag semiring, the tropical semiring, and the fuzzy semiring.
\suggestion{}{
Further, for the bag semiring, we show that computing the consistent answers of any self-join free conjunctive query whose attack graph has a strong cycle is not only $\NP$-hard but also it is $\NP$-hard to even approximate the consistent answers with a constant relative approximation guarantee.
}
\end{abstract}







\makeatletter                   
\def\mdseries@tt{m}             
\makeatother                    

\maketitle

\thispagestyle{empty}
\pagestyle{plain}

\allowdisplaybreaks
\section{Introduction and summary of results}

Database repairs and the consistent answers of queries provide a principled approach to coping with inconsistent databases, i.e., databases that violate one or more integrity constraints in a given set $\Sigma$.  This area of research started with the influential work by Arenas, Bertossi, and Chomicki \cite{ArenasBC99} and since then has had  
 a steady presence in database theory. Intuitively, a \emph{repair} of an inconsistent database $\DB$ is a consistent database $\DB'$ (i.e., $\DB'$  satisfies every constraint in $\Sigma$) and   differs from $\DB$ in a ``minimal'' way.  By definition, the \emph{consistent answers} $\cons(q,\Sigma,\DB)$ of a query $q$ over an inconsistent database $\DB$ is the intersection of the evaluations $q(\DB')$ of $q$ over all repairs $\DB'$ of $\DB$. Thus, every set $\Sigma$ of integrity constraints and every query $q$ give rise to the following algorithmic problem $\cons(q,\Sigma)$: given a database $\DB$, compute
 $\cons(q,\Sigma,\DB)$.
 
Since, in general, an inconsistent database may have a multitude of different repairs, computing the consistent answers can be an intractable problem. In fact, computing the consistent answers can be a \coNP-hard problem, even for conjunctive queries $q$ and for key constraints. This state of affairs motivated a series of investigations aiming to delineate the boundary between tractability and intractability in the computation of the consistent answers.
One of the most striking results along these lines is a \emph{trichotomy} theorem obtained by Koutris and Wijsen \cite{KoutrisW17}. Specifically, Koutris and Wijsen  showed that for a self-join free conjunctive query $q$ 
and a set $\Sigma$ of key constraints with one key constraint per relation of $q$,   the problem $\cons(q,\Sigma)$ exhibits one of the following three behaviours:  $\cons(q,\Sigma)$ is first-order rewritable, or it is polynomial-time computable but it is not first-order rewritable, or it is \coNP-complete.  Extending this trichotomy theorem to more expressive classes of queries and to richer types of integrity constraints has been the topic of active investigations during the past several years \cite{DBLP:conf/pods/KoutrisW18,DBLP:conf/pods/KoutrisW20,DBLP:conf/pods/KoutrisOW21,DBLP:journals/pacmmod/KoutrisOW24}.

A different direction in database research has focused on  $\mathbb K$-databases, i.e., databases in which the tuples in the relations are annotated with values from some fixed semiring ${\mathbb K}=(K,+, \times, 0,1)$. Ordinary  databases correspond to databases over the Boolean semiring ${\mathbb B}=(\{0,1\}, \vee, \wedge, 0, 1)$, while bag databases correspond to databases over the semiring ${\mathbb N}=(N,+, \times, 0,1)$ of the non-negative integers. 
The catalyst for this investigation was the paper by 
Green, Karvounarakis, and Tannen \cite{GreenKT07}, which developed 
a powerful framework for data provenance based on semirings of  polynomials. While the original framework \cite{GreenKT07} applied only to the provenance of queries expressible in negation-free first-order logic,  subsequent investigations extended the study of provenance to richer languages, including full first-order logic \cite{gradel2017semiring}, Datalog \cite{DBLP:conf/icdt/DeutchMRT14}, and least fixed-point logic \cite{DannertGNT21}.  Furthermore, 
several other topics in database theory have been examined in the context of semiring semantics, including the conjunctive query containment problem
\cite{DBLP:journals/mst/Green11, DBLP:conf/pods/KostylevRS12},
the evaluation of Datalog queries
\cite{DBLP:conf/pods/Khamis0PSW22,DBLP:journals/pacmmod/ZhaoDKRT24}, and the interplay between local consistency and global consistency for relations over semirings \cite{DBLP:journals/pacmmod/AtseriasK24}.

In this paper, we embark on an investigation of the consistent answers of queries under semiring semantics. 
Our main focus is on conjunctive queries with key constraints and on the rewritability of the certain answers of such queries.  
The first task is to address the following conceptual questions:
How should the consistent answers of queries under semiring semantics be defined? What does it mean to say that the consistent answers of a query under semiring semantics are rewritable in first-order logic? To
simplify the exposition, let us assume that the queries considered are \emph{\closed}, i.e., they have arity zero or, equivalently, the formulas that define them have no free variables. On ordinary databases, \closed~queries are called Boolean queries because they take value $0$ or $1$, but on $\mathbb K$-databases they can take any value in the universe $K$ of the semiring $\mathbb K$. 

To define the notion of the consistent answers of a query under semiring semantics, we first need to define the notion of a repair of a $\mathbb K$-database with respect to a set $\Sigma$ of key constraints, where $\mathbb K$ is a fixed semiring.  If $\DB$ is a $\mathbb K$-database, then the \emph{support} of $\DB$ is the ordinary database $\suppo(\DB)$ obtained from $\DB$ by setting value $1$ to every tuple in one of the relations of $\DB$ that has a non-zero annotation. 
In Section \ref{sec:query}, we argue that it is natural to define a \emph{repair} of a $\mathbb K$-database $\DB$ to be a  maximal 
sub-database $\DB'$ of $\DB$ whose support $\suppo(\DB')$ satisfies the key constraints at hand.
Let $q$ be a \closed~query with a set $\Sigma$ of key constraints. If $\DB$ is a $\mathbb K$-database, then on every repair $\DB'$ of $\DB$, the query returns a value $q(\DB')$ from the universe $K$ of the semiring $\mathbb K$.  We  define the \emph{consistent answers $\m\CA_{\mathbb K}(q,\Sigma, \DB)$ of $q$ on $\DB$} to be the minimum of the values $q(\DB')$, as $\DB'$ varies over all repairs of $\DB$. For this definition to be meaningful, we need to assume a total order on the universe $K$ of $\mathbb K$. Thus, we define the consistent answers of queries for \emph{naturally ordered positive semirings}, i.e., positive semirings $\mathbb K$ in which the \emph{natural} preorder of the semiring is a  total order (the precise definitions are given in the next section). The Boolean semiring, the bag semiring, the tropical semiring, and the fuzzy semiring are some of the main examples of naturally ordered positive semirings. \looseness=-1

Before introducing the notion of first-order rewritability for semirings,  let  $\qpath$ be the \closed~conjunctive query   $\exists x \exists y \exists z \, (R(x;y) \land S(y;z))$, where the semicolon separates the key positions from the non-key ones, i.e., the first attribute of $R$ is the key of $R$ and the first attribute of $S$ is the key of $S$. Let $\Sigma$ be the set consisting of these two key constraints.
Fuxman and Miller \cite{DBLP:journals/jcss/FuxmanM07} showed that the consistent answers of $\qpath$ on ordinary databases are rewritable in first-order logic. Specifically, they showed that for every ordinary database $\DB$,  $$\cons(q,\Sigma,\DB)=1 \quad \mbox{if and only if} \quad \DB \models \exists x \exists z' (R(x;z') \land \forall z( R(x;z) \to \exists y S(z;y))).$$
Thus, $\cons(q,\Sigma,\DB)$ can be computed by a single evaluation of  a first-order sentence on $\DB$ and without  evaluating repeatedly $q(\DB')$, as $\DB'$ ranges over the potentially exponentially many repairs of $\DB$. 
Let $\K=(K,+,\times, 0,1)$ be a naturally ordered positive semiring. 
In Section \ref{sec:mca}, we show that for every $\K$-database $\DB$, the following holds for the consistent answers $\m\CA_\K(\qpath,\Sigma,\DB)$:
$$\m\CA_\K(\qpath,\DB) = 
\sum_{a\in D} \min_{b\in D:R^\DB(a,b)\neq 0}({R^\DB(a,b) \times \min_{c\in D:S^\DB(b,c)\neq 0}{S^\DB(b,c)}}),
 $$
 where $D$ is the active domain of  $\DB$. Thus,  $\m\CA_\K(\qpath,\Sigma, \DB)$ can be evaluated directly on $\DB$ and without considering the repairs of $\DB$. 
 
 Motivated by the above properties
 of $\qpath$, we  introduce the logic $\Logic_\K$, which is a variant of first-order logic with a minimization operator $\nabla$ and a limited form of negation ($\supp$) that   flattens non-zero annotations to zero. We give rigorous semantics to the formulas of  the logic $\Logic_\K$ on every naturally ordered positive semiring and then investigate
 when the consistent answers of conjunctive queries are rewritable in this logic.
 
 Let $q$ be a self-join free \closed~conjunctive query with one key constraint per relation.
 Our main result asserts that the consistent answers of $q$ are rewritable in the logic
 $\Logic_\K$ if and only if the \emph{attack} graph of $q$ is acyclic. 
This result generalizes an analogous result of Koutris and Wijsen \cite{KoutrisW17} for the Boolean semiring, but also yields new results for a multitude of semirings, including the bag semiring, the tropical semiring, and the fuzzy semiring.
The notion of the attack graph was introduced by Wijsen~\cite{DBLP:conf/pods/Wijsen10,DBLP:journals/tods/Wijsen12} 
and has played an important role in the study of the consistent answers of self-join free conjunctive queries on ordinary databases.  Here, we leverage the insights obtained in this earlier study, but also obtain new insights that entail a further analysis of the properties of the attack graph when the query is evaluated under semiring semantics. As an illustration, it will turn out that the consistent answers  $\m\CA_\K(\qpath,\Sigma, \DB)$ of the query $\qpath$ on $\DB$ are definable by the 
 $\Logic_\K$-formula
 $
\exists x \nabla_{R(x,y)} y. \,({R(x,y) \times \nabla_{S(y,z)} z .\, {S(y,z)}})$. 

Let $q$ be a fixed self-join free query with one key per relation.
Koutris and Wijsen \cite{KoutrisW17} showed that if  the attack graph of $q$ contains a \emph{strong} cycle, then computing the consistent answers $\cons(q,\Sigma, \DB)$ of~$q$ on ordinary databases $\DB$ is a \coNP-complete problem \suggestion{}{(this is a data complexity result, since $q$ is fixed)}. 
Note that if $\K$ is a naturally ordered positive semiring, then computing  $\m\CA_\K(q,\DB)$ on  $\K$-databases $\DB$ is an optimization problem. Here, we focus on the bag semiring ${\mathbb N}=(N,+,\times, 0,1)$ and show that if the attack graph of $q$ contains a strong cycle, then computing 
$\m\CA_{\mathbb N}(q,\Sigma, \DB)$ on bag databases $\DB$ not only is a 
\NP-hard problem but also  it is  \NP-hard to even approximate $\m\CA_{\mathbb N}(q,\Sigma, \DB)$ with a constant relative approximation guarantee.
 This result paves the way to expand the study of  the consistent answers of queries on annotated databases with methods and techniques from the rich theory of approximation algorithms for optimization problems.

\commentout{
This was proved via a reduction from $\cons(\qsink,\DB)$, where
$\qsink$ is the query   $\exists x \exists y \exists z \, (R(x;z) \land S(y;z))$; the consistent answers of $\qsink$ were already known to be \coNP-complete \cite{DBLP:journals/ipl/KolaitisP12}. Note that if $\K$ is a naturally ordered positive semiring, then computing the consistent answers $\m\CA_\K(\qsink,\DB)$ on  $\K$-databases $\DB$ is an optimization problem. Here, we focus on the bag semiring ${\mathbb N}=(N,+,\times, 0,1)$ and show that computing 
$\m\CA_{\mathbb N}(\qsink,\DB)$ on bag databases $\DB$ not only is a 
\coNP-hard problem but also does not have a polynomial-time constant-approximation algorithm, unless $\PTIME = \NP$. This result paves the way to expand the study of  the consistent answers of queries on annotated databases with methods and techniques from the rich theory of approximation algorithms for optimization problems.
}

\section{Preliminaries}\label{sec:preliminaries}

\paragraph{Semirings.}
A \emph{commutative semiring} is an algebraic structure $\K=(K, +, \times, 0, 1)$ with $0 \neq 1$ and such that $(K, +,0)$ and $(K,\times, 1)$ are commutative monoids, $\times$ distributes over $+$, and $0\times a = a \times 0 = 0$ for every $a \in K$. A semiring has no \emph{zero-divisors} if $a\times b = 0$ implies that  $a=0$ or $b=0$.
We say that $\K$ is \emph{positive} if $a+b = 0$ implies $a=0$ and $b=0$, and also $\K$ has no zero-divisors.
A semiring is \emph{naturally ordered} if the canonical preorder $\leq_{\K}$, defined by $a \leq_{\K} b \Longleftrightarrow \exists c~(a+c = b)$, is a total order.
Every non-empty finite set $A$ of elements from a naturally ordered semiring has a minimum element, denoted by $\min(A)$. 
%
Unless specified otherwise, from now on we assume  $\K$ to be a naturally ordered positive semiring. Examples of such semirings include the
\emph{Boolean} semiring ${\mathbb B}=(\{0,1\}, \lor, \wedge, 0, 1)$, the \emph{bag} semiring
${\mathbb N}=(N,+, \cdot, 0, 1)$ of the natural numbers, the \emph{tropical} semiring
${\mathbb T}=([0,\infty], \min, +, \infty, 0)$, the \emph{Viterbi} semiring ${\mathbb V}=([0,1], \max,\times, 0,1)$, and the \emph{fuzzy} semiring ${\mathbb F}=([0,1], \max, \min, 0,1)$.
If $\mathbb K$ is  the Boolean or the bag or the Viterbi or the fuzzy semiring, then $\leq_\K$ coincides with the standard order on the universe of $\K$, while if $\K$ is the tropical semiring, then $\leq_\K$ is the reverse of the standard order on the universe of $\K$.

\paragraph{$\K$-relations and $\K$-databases.}
A \emph{(relational) schema} $\tau$ is a finite set of relation symbols each with a positive integer as its arity.
We fix a countable set $\dv$ of possible data values.
An $n$-ary \emph{$\K$-relation}, where $n>0$,  is a function $R \colon \dv^n \to K$ such that $R(t)=0$ for all but finitely many $n$-tuples of~$\dv^n$.  
The \emph{support} of $R$ is defined as $\suppo(R) \dfn \{ t \in \dv^n \colon R(t) \neq 0 \}$.
If $R$ and $T$ are $\K$-relations of the same arity, we write $R \leq_\K T$ if $R(\vec{a}) \leq_\K T(\vec{a})$ for all $\vec{a}$ in the support of $R$. We write $R \subseteq T$, if $\suppo(R)\subseteq \suppo(T)$ and  $R(\vec{a}) = T(\vec{a})$ for all $\vec{a}$ in the support of $R$.
A \emph{$\K$-database} $\DB$ over a schema~$\tau$ is a collection of $\K$-relations, that is, a collection of functions $R_i: \dv^n \to K$, such that the arities of $R_i$ match that of the corresponding relation symbols in $\tau$.
We often write $R^\DB$ to denote the interpretation of the relation symbol R in $\DB$.
The \emph{support} of $\DB$,  written  $\suppo(\DB)$, is the database that consists of the supports $\suppo(R)$ of the $\K$-relations $R$ of $\DB$.
For $\K$-databases $\DB$ and $\DB'$ of the same schema $\tau$, we write $\DB \leq_\K \DB'$ ($\DB \subseteq \DB'$, resp.) if $R^\DB \leq_\K R^{\DB'}$ ($R^\DB \subseteq R^{\DB'}$, resp.) for every $R\in \tau$.
The \emph{active domain} of the $\K$-database $\DB$, denoted by $D$, is the set of all data values that occur in the support of some $\K$-relation of $\DB$. \looseness=-1
\suggestion{}{In this work, we only consider $\K$-databases with a non-empty active domain. This means that there is at least one relation $R$ and a tuple $t$ such that $R(t)\neq 0$. 
} 

\paragraph{Conjunctive Queries}

We assume familiarity with basic definitions and notions related to first-order logic \FO in the context of relational databases (e.g., see \cite{EFT84, abi95}).
We fix a countably infinite set $\Var$ of \emph{first-order variables} $x, y, x_1,\dots, x_n$ etc. and write $\vec{x}$ to denote a \emph{tuple} of variables.
We write $\var(\vec{x})$ to denote the set of variables that occur in $\vec{x}$.  
An \emph{assignment} on a $\K$-database  $\DB$ is a total function $\alpha \colon \Var \to D$. 
For a variable $x\in \Var$ and a value $a\in D$, we write $\alpha(a/x)$ to denote the assignment 
that maps $x$ to $a$, and otherwise agrees with $\alpha$.
%

A \emph{conjunctive query} (CQ) is an \FO-formula of the form
$q(\vec{x}) \dfn \exists \vec{y} \big(R_1(\vec{z}_1) \land \ldots \land R_n(\vec{z}_n)\big)$,
where each $R_i(\vec{z}_i)$ is a relational atom and each $\vec{z}_i$ is a tuple of variables in $\vec{x}$ and $\vec{y}$. 
We assume that  all quantified variables of $q$ occur in the quantifier-free part of the query.
We say that a CQ $q$ is \emph{closed} if $q$ has no free variables; otherwise, we say that
$q$ is \emph{open}.
We write $\hat{q}$ to denote the quantifier-free part of $q$.
We call $q$ \emph{self-join-free} (\sjfCQ) if no relation symbol occurs more than once in $\hat{q}$.
%
If $\alpha$ is an assignment, we write $\alpha(\hat{q})$ for the set $\{R(\alpha(\vec{x})) \mid R(\vec{x}) \text{ is a conjunct of $\hat{q}$}\}$ of facts of $\hat{q}$ under $\alpha$.

Following  \cite{GreenKT07},
if $q(\vec{x}) = \exists \vec{y}(R_1(\vec{z}_1) \land \ldots \land R_k(\vec{z}_k))$ is a CQ,
$\DB$ is a $\K$-database, and $\alpha$ is an  assignment on $\DB$,  
then, if $\sum$ and $\times$ are the (iterated) addition and multiplication of the semiring $\K$, 
the semantics of $q$ on $\DB,\alpha$  is the semiring value
\[
\suggestion{}{
\eval{q}{\DB}{\alpha} \dfn \sum_{\vec{a}\in D^{\vert \vec{y}\rvert}}  R_1^{\DB}(\beta(\vec{z}_1)) \times \ldots \times R_k^{\DB}(\beta(\vec{z}_k)), \quad\text{ where } \beta\dfn \alpha(\vec{a}/\vec{y}).
}
\]

\section{Repairs and consistent answers under semiring semantics} \label{sec:query}

Integrity constraints are semantic restrictions that the data of interest must obey. Integrity constraints on ordinary databases are typically expressed as sentences of first-order logic. In particular, this holds true for key constraints,  the most widely used integrity constraints. A \emph{key constraint} on a relation symbol $R$
asserts that the values of some  attributes of $R$ determine the values of all other attributes of $R$. In what follows, we adopt the convention that the key attributes occupy the leftmost positions in the relation symbol $R$. We will write $R(\vec{x}; \vec{y})$ 
to denote that the attributes in $\vec{x}$ form a key of $R$, and set $\key{R}\dfn \var(\vec{x})$.
Clearly, every key constraint is expressible in  \FO. For example, if we have a ternary relation symbol $R(x_1,x_2;y)$, then the first-order sentence $\forall x_1 \forall  x_2 \forall y \forall z (R(x_1,x_2,y) \land R(x_1,x_2,z)\rightarrow y=z)$ tells that the first two attributes of $R$ form a key. 

Let $\K$ be a semiring and let $\DB$ be a $\K$-database. What does it mean to say that $\DB$ satisfies a key constraint or, more generally, an integrity constraint $\psi$, where $\psi$ is an \FO-sentence? To answer this question, we have to give semiring semantics to \FO. In the previous section, we already gave such semantics to the fragment of \FO that expresses conjunctive queries, where  the addition and the multiplication operations of $\K$ were used to interpret, respectively, existential quantification and conjunction. This approach extends naturally to the negation-free fragment of \FO, but more care is needed to assign semiring semantics to arbitrary \FO-formulas. As part of the study of provenance in databases, Gr\"adel and Tannen \cite{gradel2017semiring} gave semiring semantics to \FO-formulas in negation normal form NNF (i.e., all negation symbols are ``pushed'' to the atoms) by using the notion of an \emph{interpretation}, which is a function that assigns semiring values to atomic or negated atomic facts. Here, we give  semiring semantics to \FO-formulas on a $\K$-database $\DB$ by, in effect, considering a particular \emph{canonical} interpretation on $\DB$; a similar approach was adopted by Barlag et al. in \cite{BHKPV23unified} for $\K$-teams, which can be viewed as  $\K$-databases over a schema with a single relation symbol. Appendix \ref{A:truth} contains a more detailed discussion of semiring semantics via interpretations.

Let $\mathbb K$ be a  semiring, $\DB$  a $\K$-database, and $\alpha\colon\Var\to D$ an assignment. If $\varphi$ is an \FO-formula in negation normal form, then we define the semiring value
$\eval{\varphi(\vec{x})}{\DB}{\alpha}$  recursively as follows:
\label{def:FOsemantics}
\begin{align*}
     \eval{R(\vec{x})}{\DB}{\alpha}&= R^\DB(\alpha(\vec{x})) 
     &    &    \\
     \eval{\neg R(\vec{x})}{\DB}{\alpha} &=  
     \begin{cases}
     1 &\hspace{-2mm}\text{if } \eval{R(\vec{x})}{\DB}{\alpha}=0\\
     0 &\hspace{-2mm}\text{if } \eval{R(\vec{x})}{\DB}{\alpha}\neq 0
     \end{cases} &
       \eval{(x * y)}{\DB}{\alpha} &=  
     \begin{cases}
     1 &\hspace{-2mm}\text{if }\alpha(\vec{x})*\alpha(\vec{y})\\
     0 &\hspace{-2mm}\text{otherwise}
     \end{cases} \mbox{where}~ *\in \{=,\neq\}\\
    \eval{(\varphi \land \psi)}{\DB}{\alpha}  & = \eval{\varphi}{\DB}{\alpha} \times \eval{ \psi}{\DB}{\alpha}
     &   \eval{(\varphi \lor \psi)}{\DB}{\alpha} & = \eval{\varphi}{\DB}{\alpha} + \eval{ \psi}{\DB}{\alpha}  \\
    \eval{\forall x\varphi}{\DB}{\alpha} & =    \prod_{a \in D} \eval{\varphi }{\DB}{\alpha(a/x)} 
    &  \eval{\exists x\varphi}{\DB}{\alpha} & =    \sum_{a \in D} \eval{\varphi }{\DB}{\alpha(a/x)}.
  \end{align*}

It is straightforward to prove by induction that if $\alpha$ and $\beta$ are assignments that agree on the free variables of an \FO-formula $\varphi$, then  $\eval{\varphi}{\DB}{\alpha}=\eval{\varphi}{\DB}{\beta}$, for every  $\K$ database $\DB$. 
Thus, from now on, when we consider an \FO-formula $\varphi$ and an assignment $\alpha$, we will assume that $\alpha$ is a function defined on the free variables on $\varphi$, hence $\alpha$ is a finite object. 
Furthermore, if $\varphi$ is a sentence (i.e., $\varphi$ has no free variables), then the semiring value $\eval{\varphi}{\DB}{\alpha}$
does not depend on the assignment $\alpha$.  In what follows, we write $\varphi (\DB) $ to denote that value, if $\varphi$ is an \FO-sentence.

The following result is a consequence of Proposition 9 in \cite{gradel2017semiring}.

\begin{proposition} \label{pro:interp}
If $\varphi$ is an \FO-formula in NNF,  $\DB$ is a $\K$-database, and $\alpha$ is an assignment, then 
\begin{equation}\label{eq:truth}
\eval{\varphi}{\DB}{\alpha} \not = 0 \quad\text{ if and only if }\quad \suppo(\DB), \alpha \models \varphi,
\end{equation}
where the symbol $\models$ in the right-hand side refers to satisfaction in standard (set-based) \FO.  In particular, if $\varphi$ is an \FO-sentence in NNF, then $\varphi(\DB) \not = 0$ if and only if ~ $\suppo(\DB) \models \varphi$.
\end{proposition}

Proposition \ref{pro:interp} leads to the following natural definition for satisfaction over $\K$-databases.
\begin{definition}\label{def:satisfaction}
 Let $\varphi$ be an \FO-formula in NNF, $\DB$  a $\K$-database, and $\alpha$ an assignment. We say that $\DB,\alpha$ \emph{satisfies} $\varphi$, denoted  $\DB,\alpha \models_\K \varphi$, if  
    $\eval{\varphi}{\DB}{\alpha}\neq 0$ (equivalently, $\suppo(\DB), \alpha \models \varphi$). In particular, if
    $\varphi$ is an \FO-sentence in NNF, then
    we write $\DB \models_{\K} \varphi$ if
$\varphi(\DB) \not = 0$ (equivalently,
$\suppo(\DB)\models \varphi$).

If $\Sigma$ is a set of \FO-sentences in NNF,  we write $\DB \models_\K \Sigma$ to denote that $\DB\models_\K \varphi$ for every $\varphi \in \Sigma$. 
\end{definition}

With Definition \ref{def:satisfaction} at hand, we are now ready to define the notion of a repair of a $\K$-database $\DB$ with respect to a set $\Sigma$ of key constraints. Here, we will assume that every key constraint is defined by an \FO-sentence in NNF in the standard way.  For example, if $R(x_1,x_2;y)$ is a relation symbol in which the first two attributes form a key, then the key constraint expressing this property will be given by the \FO-sentence 
$\forall x_1,x_2,y,z(\neg R(x_1,x_2,y) \lor \neg R(x_1,x_2,z) \lor y = z)$.

Recall also that if $\DB$ and $\DB'$ are two $\K$-databases, then $\DB' \subseteq \DB$ means that for every relation symbol $R$, the following two properties hold: (i) $\suppo(R^{\DB'})\subseteq \suppo(R^{\DB})$;  (ii) $R^{\DB'}(\vec{a})=R^{\DB}(\vec{a})$, for every $\vec{a} \in \suppo (R^{\DB'})$. We will write $\DB' \subset \DB$ to denote that $\DB'\subseteq \DB$ and $\DB'\not = \DB$.

\begin{definition}\label{def:repair_Kdb_egds}
    Let $\Sigma$ be a set of key constraints, $\K$ a naturally ordered positive semiring, and $\DB$ a $\K$-database. A $\K$-database $\DB'$ is a \emph{repair} of $\DB$ w.r.t.
    $\Sigma$ if:
    \begin{enumerate}
        \item $\DB' \models_\K \Sigma$ (which amounts to  $\suppo(\DB') \models  \varphi$, for every $\varphi \in \Sigma$); 
        \item $\DB'\subseteq \DB$ and there is no $\K$-database $\DB''$ such that $\DB' \subset \DB''\subseteq \DB$ and $\DB''\models_\K \Sigma$. %
    \end{enumerate}
    We write $\Rep(\DB,\Sigma)$ for the collection of repairs of $\DB$ with respect to $\Sigma$.
\end{definition}

Several remarks are in order now concerning the notion of repair in Definition \ref{def:repair_Kdb_egds}. First, 
in the case of ordinary databases (i.e.,  $\K$ is the Boolean semiring $\mathbb B$), this notion  coincides with the standard notion of a (subset) repair of an ordinary database with respect to a set of key constraints.
Second, this notion is quite robust in the sense that we get the same notion if, instead of $\subseteq$,  we had used the more relaxed notion $\leq_{\K}$ for comparing $\K$-databases. Finally, the notion can be extended by considering arbitrary \FO-sentences as integrity constraints. Here, we gave the definition of a repair with respect to key constraints, since, in this paper, our entire focus is on such constraints.

\begin{example} Let ${\mathbb N}=(N,+,\cdot, 0,1)$ be the bag semiring and consider a schema consisting of a single ternary relation  $R(x_1,x_2;y)$  in which the first two attributes form a key. Let $\DB$ be the $\mathbb N$-database  with $R^\DB(a,b,c)=2, R^\DB(a,b,d)=3, R^\DB(a,a,a) = 4$, and  $R^\DB(a',b',c')=0$ for all other values (the values $a,b,c,d$ are assumed to be distinct). Then $\DB$ has exactly two repairs $\DB_1$ and $\DB_2$ with respect the key constraint of $R$, where the non-zero values of $R^{\DB_1}$
are $R^{\DB_1}(a,b,c)=2,  R^{\DB_1}(a,a,a) = 4$, while the non-zero values of $R^{\DB_2}$ are
$R^{\DB_2}(a,b,d)=3 , R^{\DB_2}(a,a,a) = 4$.
\end{example}

We are now ready to define the notion of consistent answers in the semiring setting.

\begin{definition} \label{def:cqa}
Let $\K$ be a naturally ordered positive semiring,  $\Sigma$ a set of key constraints, $\varphi$ an \FO-formula, $\DB$ a $\K$-database, and  $\alpha$ a variable assignment. 
%
The \emph{consistent answers $\mathrm{m}\CA_\K(\varphi, \Sigma, \DB, \alpha)$ of $\varphi$ on $\DB, \alpha$ with respect to $\Sigma$} is defined as
\(
\min_{\DB'\in \Rep(\DB,\Sigma)} \varphi(\DB',\alpha).
\)
%
If $\varphi$ is an \FO-sentence, then 
the \emph{consistent answers 
of $\varphi$ on $\DB$ with respect to $\Sigma$} is the value
$\mathrm{m}\CA_\K(\varphi, \Sigma, \DB) = \min_{\DB'\in \Rep(\DB,\Sigma)} \varphi(\DB')$.  
\end{definition}
The consistent answers 
provide
the tightest lower 
bound on the values $\varphi(\DB',\alpha)$ as $\DB'$ ranges over all repairs of $\DB$.  
On the Boolean semiring $\mathbb B$, they 
coincide with the consistent answers.

 


\begin{example} Let ${\mathbb V}=([0,1], \max, \times, 0, 1)$ be the Viterbi semiring. For   $n\geq 2$, let $\DB$ be a $\mathbb V$-database with $\mathbb V$-relations $E_i(x;y)$, for $1\leq i\leq n$, each encoding a simple directed edge-weighted graph in which
the weight of every edge is a  real number in $(0,1]$.
The weight of an edge can be thought of as the \emph{confidence} of the edge. We define the \emph{confidence of a path along $E_1,E_2,\dots, E_n$} to be the product of the confidences of its edges. Let us now consider the closed conjunctive query
$q_n \dfn \exists x_0\cdots \exists x_{n} \big(E_1(x_0;x_1)\land \dots \land E_n(x_{n-1};x_n) \big)$  and  the set $\Sigma$ of the key constraints of $E_i$, $1\leq i \leq n$. 
Then, 
$\m\CA_{\mathbb V}(q_n, \Sigma, \DB, \alpha)$
is a number $c \in [0,1]$ that has the following two properties: (i) in every repair of $\DB$, there is a path along $E_1,E_2,\ldots,E_n$ of confidence at least $c$; (ii) there is a repair of $\DB$ in which the maximum confidence of a path along $E_1,\ldots,E_n$ is $c$. 
\end{example}

Henceforth,  we will focus on consistent answers for self-join-free CQs and key constraints.

\section{Consistent answers and query rewriting}\label{sec:mca}

Data complexity is one of the most useful measures in studying query answering \cite{Vardi82thecomplexity}. This stems from the observation that in practice queries are typically of small size (e.g., they are written by a user), while databases are of large size. In data complexity, the query is fixed and only the database is the input. Hence, each query $q$ gives rise to a separate computational problem.

\begin{definition}
Let $\K$ be a naturally ordered positive semiring, $q$ a CQ, and $\Sigma$ a set of key constraints.
We define the following function problem.
    \begin{center}
        \fbox{\begin{minipage}{30em}
            \textsc{Problem}: $\m\CA_\K(q,\Sigma)$
        
            \textsc{Input}: A $\K$-database $\DB$ and an assignment $\alpha$.  
            
            \textsc{Output}: The value of $\m\CA_\K(q,\Sigma,\DB,\alpha)$.
        \end{minipage}}
    \end{center}
\end{definition}

\suggestion{%
}
{If $\K=\B$ and $q$ is a closed CQ, then $\m\CA_\K(q,\Sigma)$ is a decision problem, which in the introduction we  denoted by $\CONS(q,\Sigma)$.
We say that $\CONS(q,\Sigma)$ is \emph{$\FO$-rewritable} if there is a  $q'\in\FO$ such that $\cons(q,\Sigma, \DB) = q'(\DB)$, for every ~$\DB$.
}%
%
%
The main benefit of an \FO-rewriting  of $\m\CA_\B(q,\Sigma)$  is that such a rewriting reduces consistent query answering to query evaluation.  Thus, a database engine developed for query evaluation  can  be directly used to compute the consistent answers. Note that having a query rewriting in some logic is trivial; for instance, one may internalize the process of checking all repairs using second-order logic.
%
The real benefit of  having an \FO-rewriting is that the evaluation of \FO-formulas, with respect to data complexity, is fast and highly parallelizable. In fact, $\FO$-definable properties can be recognized by circuits of constant-depth and polynomial-size;  more precisely, they lie in the circuit complexity class $\DLOGTIME$-uniform $\aco$ \cite{DBLP:journals/jcss/BarringtonIS90}.
Our goal is to identify a suitable  notion of $\FO$-rewritability for $\m\CA_\K(q,\Sigma)$, where $\K$ is a naturally ordered positive semiring.

As a motivation, we first present an example of the rewriting of $\m\CA_\K(q,\Sigma)$ in the  semiring context for a particular query $q$. In Section \ref{sec:logic}, we introduce the logic $\Logic_\K$ that will be used to express rewritings. In Section \ref{sec:logic_good}, we show that the data complexity of $\Logic_\K$ lies in a semiring variant of uniform $\aco$. Finally, in Section \ref{sec:rewriting}, we present our main result concerning the rewriting of  $\m\CA_\K(q,\Sigma)$ for sjfCQs and key constraints, one key  for each relation in $q$.  In what follows, we consider only sjfCQs and key constraints that can be read from the query; thus, we drop $\Sigma$ from the notation and write simply $\m\CA_\K(q)$, $\m\CA_\K(q,\DB,\alpha)$, and $\Rep(\DB)$. Recall that we also drop  $\alpha$, when we consider closed queries and when the proofs do no technically require an assignment.

%
%

\subsection{Rewriting of the consistent answers of the path query}\label{sec:qpath}

Recall that $\qpath = \exists x \exists y \exists z( R(x;y)\land S(y;z))$, for which $\m\CA_\B(\qpath)$  has the following
\FO-rewriting:
\(
\exists x \exists z' (R(x,z') \land \forall z( R(x,z) \to \exists y S(z,y)))
\)~\cite{DBLP:journals/jcss/FuxmanM07}.
We want to obtain a similar expression in our setting for \qpath, that is, an expression using semiring operations that provides the answer $\m\CA_\K(\qpath,\DB)$ for $\K$-databases without having to evaluate \qpath on every  repair of $\DB$.
We define the expression $\epath$ as follows:
\begin{equation}\label{eq:epath}
     \epath \dfn \sum_{a\in D} \min_{b\in D:R^\DB(a,b)\neq 0}({R^\DB(a,b) \times \min_{c\in D:S^\DB(b,c)\neq 0}{S^\DB(b,c)}}),
\end{equation}
 where $\sum$, $\times$, and $\min$ refer to operations of a naturally ordered positive semiring $\K$.\suggestion{}{
 We call $R^\DB(a,b)\neq 0$ and $S^\DB(b,c)\neq 0$ \emph{guards} of the minimisation operators.
 The $\min$-operator in~\eqref{eq:epath} may take an empty set as its range. This occurs specifically when there is no~$b\in D$ such that $R^\DB(a,b)\neq 0$.
 In expressions of the above form, we adopt the convention that $\min(\emptyset)=0$.} 
%
%
We claim that
\begin{align}
\m\CA_\K(\qpath,\DB)
= \epath,
\label{eq:rewriting_qpath}
\end{align}
for every $\K$-database $\DB$. 
%
%
%
Next we show that the equality holds.
Note that 
    \[
    \m\CA_\K(\qpath,\DB) =
    \min_{\DB' \in \Rep(\DB)} \FOKeval{\qpath}{\DB'} =
    \min_{\DB' \in \Rep(\DB)} \sum_{a,b,c\in D'} R^{\DB'}(a,b) \times S^{\DB'}(b,c),
    \]
where the former equality is the definition of $\m\CA_\K$ and the latter follows from the semiring semantics for CQs.  
%
To show (\ref{eq:rewriting_qpath}), 
it suffices to establish that the following two statements hold:
\begin{enumerate}[label=(\roman*)]
    \item For every repair $\DB'$ of $\DB$, we have that $\epath \leq_\K \qpath(\DB')$; \label{qpath:first}
    \item There is a repair $\DB^*$ of $\DB$ such that $\epath = \qpath(\DB^*)$. \label{qpath:second}
\end{enumerate}

The first statement is proved by inspecting the expression case-by-case, while the second is proved by constructing a suitable repair. The full proof is in Appendix~\ref{appendix_qpath}.

\suggestion{}{
In order to construct a rewriting of $\m\CA_\K(\qpath)$, we need a logical formula that evaluates to $\epath$, for every $\K$-database $\DB$. 
We first argue why existing rewritings for the Boolean semiring~\cite{DBLP:journals/jcss/FuxmanM07, KoutrisW17} do not straightforwardly extend to arbitrary naturally ordered positive semirings.  
Intuitively, this is because expressions like~\eqref{eq:epath} are built from sum ($\sum$), product $(\prod)$, and $\min$. While sum and product naturally correspond to existential and universal quantification, respectively, standard rewritings do not provide a natural equivalent for $\min$. In fact, $\min$ is not needed in the Boolean semiring, since in expressions like~\eqref{eq:epath}, the guarded $\min$-operator would always evaluate to~$1$ provided that the guards are supported, and to~$0$ otherwise. Unguarded $\min$-operators, on the other hand, directly correspond to product on the Boolean semiring.

To formalize the preceding discussion, the first-order rewriting of Fuxman and Miller~\cite{DBLP:journals/jcss/FuxmanM07} and Koutris and Wijsen~\cite{KoutrisW17} for $\m\CA_\B(\qpath)$,  when expressed in negation normal form, is 
\begin{equation}\label{eq:fuxman}
\exists x \exists z' (R(x,z') \land \forall z( \neg R(x,z) \lor \exists y S(z,y))).
\end{equation}
Under the semiring semantics of first-order logic, this rewriting would evaluate to the expression
\begin{equation}\label{eq:fuxmannew}
\sum_{a\in D, b\in D} \Big(
R^\DB(a,b) \times
\prod_{c\in D} \big( \supp(R^\DB(a,c)) +
\sum_{d\in D} S^\DB(c,d)\big)\Big),
\end{equation}
where $\supp(k)$ is a function introduced in Section~\ref{sec:logic} that maps any non-zero semiring value $k$ to~$0$, and $0$ to $1$.
While this expression is correct for the Boolean semiring, as expected, it is incorrect for other semirings, such as the bag semiring, since it uses a product of semiring values instead of minimization.
This raises the question of whether a correct rewriting is obtained by interpreting the universal quantifier as a minimization operator. Under this interpretation, the rewriting~\eqref{eq:fuxman} evaluates to the expression
\[
\sum_{a\in D, b\in D} \Big(
R^\DB(a,b) \times
\min_{c\in D} \big( \supp(R^\DB(a,c)) +
\sum_{d\in D} S^\DB(c,d)\big)\Big).
\]
To see why the latter expression is incorrect for the bag semiring, consider an $\N$-database $\DB$ with $R^\DB(a,b)=R^\DB(a,c)=1$ and $S^\DB(b,d)=S^\DB(c,d)=1$. Then $\m\CA_\N(\qpath,\DB)=1$, but the expression evaluates to $2$.



To obtain a rewriting whose semantics matches $\epath$, we need a logic capable of expressing the \emph{guarded} minimization operators used in \eqref{eq:epath}. A rewriting in such a logic could be expressed as follows:
\begin{equation}\label{eq:lepath}
\exists x \nabla_{R(x,y)} y. \,({R(x,y) \times \nabla_{S(y,z)} z .\, {S(y,z)}}),
\end{equation}
where $\nabla_{R(x,y)} y$ is a sort of minimization operator that utilizes a \emph{guard}. 
In the next section, we introduce a logic $\Logic_\K$ that supports such formulas.
}

\subsection{A logic for query rewriting over $\K$-databases}\label{sec:logic}

%
We have just exhibited an expression based on semiring operations for rewriting  the consistent answers of the path query on $\K$-databases.
Next we make this more formal, and introduce a general notion of query rewritability for $\K$-databases. For this purpose, we need to define a suitable logic.


%
As before, $\tau$ is a relational schema and $\K$ a naturally ordered positive semiring. 
%
%
The syntax of $\Logic_\K$ is given by the following grammar:
\begin{align*}
\varphi \,\coloneqq\,& R(\vec{x}) \,|\, x=y \,|\, \varphi\land\varphi 
\,|\, \varphi\lor\varphi 
\,|\, \exists x\, \varphi
\,|\, \nabla x \varphi(x) \,|\, \supp(\varphi),
\end{align*}
where $R\in\tau$ of arity $r$, and $\vec{x}=(x_1,\ldots,x_r)$, for $x_1,\ldots,x_r,x,y\in \Var$.

Let $\DB$ be a $\K$-database and $\alpha\colon \Var \to D$ an assignment.
The value $\FOKeval{\varphi}{\DB,\alpha}$ of a formula $\varphi\in \Logic_\K$ is defined recursively as follows. The cases for literals, Boolean connectives, and the existential quantifier are as  the  semiring semantics for $\FO$ (see page \pageref{def:FOsemantics}). \suggestion{}{The semantics of the new constructs is as follows:}
%
%
 \begin{align*}
        \FOKeval{\nabla x \varphi(x)}{\DB,\alpha} 
    & = \min_{a\in D} \FOKeval{\varphi}{\DB,\alpha(a/x)}
    & \FOKeval{\supp(\varphi)}{\DB,\alpha}
    &=     
    \begin{cases}
     1 & \text{if }\FOKeval{\varphi}{\DB,\alpha} = 0 \\
     0 & \text{otherwise}.
    \end{cases}
 \end{align*}
%
%
%
 Over the Boolean semiring, the semantics of the minimization operator $\nabla x$ coincides with the universal quantifier and $\supp(\varphi)$ corresponds to negation.
In general, the $\supp$ operator serves as a (weak) negation operator since it flattens every non-zero annotation to return $0$, thus in a sense losing all the shades of information provided by non-zero annotations. We use the shorthand $\suppo(\varphi)$ for $\supp(\supp(\varphi))$, which flattens all non-zero weights to $1$.


\begin{remark}
    $\Logic_\B$ is essentially $\FO$; furthermore, in $\Logic_\B$, $\nabla x$ is $\forall x$ and $\supp(\varphi)$ is $\neg \varphi$.
\end{remark}
This observation can be extended to the result that $\FO$ embeds into $\Logic_\K$, for every naturally ordered positive semiring $\K$, in the following sense. \suggestion{}{The translation between the logics is obtained by identifying $\nabla x$ with $\forall x$ and $\supp(\varphi)$ with $\neg \varphi$.}
\begin{proposition}\label{prop:semiring_to_Boolean}
Let $\K$ be a naturally ordered positive semiring. For every $\varphi\in \Logic_\K$, there exists $\psi\in \FO$ such that
$\FOKeval{\varphi}{\DB,\alpha} \neq 0$  if and only if  $\suppo(\DB), \alpha \models \psi$,
for every $\K$-database $\DB$ and every assignment $\alpha$.
Conversely, for every $\psi\in \FO$, there exists $\varphi\in \Logic_\K$ such that the  ``if and only if'' holds.
If the annotations of $\DB$ are from the set $\{0,1\}$, then $\FOKeval{\varphi}{\DB,\alpha} \neq 0$  if and only if  $\DB, \alpha \models~\psi$.
\end{proposition}


The formula $\nabla x.\, \varphi(x)$ computes the minimum value of $\varphi(a/x)$, where $a$ ranges over the active domain of the database. However, sometimes we want $a$ to range over the support of some definable predicate, that we call \emph{guard} and write $G$. \suggestion{}{
In general, the guard is an $\Logic_\K$-formula $G(\vec{y},z)$, and typically it has free variables from $\varphi$.} For this purpose, we define the following shorthand
\begin{equation}\label{eq:nabla}
\nabla_G z. \,\varphi(\vec{y},z)  \dfn \nabla z. \theta(\vec{y},z), \text{ where }\theta(\vec{y},z)\dfn \Big(\big( \supp(G(\vec{y},z)) \land \exists z' \varphi(\vec{y},z') \land \chi \big) \lor \big(\varphi(\vec{y},z) \land \chi\big) \Big),
\end{equation}
where $G, \varphi \in \Logic_\K$ and $\chi \dfn \suppo(\exists z' G(\vec{y},z'))$.
The idea of \eqref{eq:nabla} is as follows. 
\suggestion{}{
We first explain $\theta(\vec{y},z)$.
Informally, if the subformula $\exists z' G(\vec{y},z')$ of $\chi$ is not supported (i.e., $G(\vec{y},z')$ evaluates to~$0$ for every~$z'$), then~$\theta(\vec{y},z)$ evaluates to~$0$.
The more interesting case is where $\exists z' G(\vec{y},z')$ is supported.
In that case, $\theta(\vec{y},z)$ returns the value of $\varphi(\vec{y},z)$  if $G(\vec{y},z)$ is supported; otherwise it returns the sum of all values $\varphi(\vec{y},z')$, where $z'$ ranges over the active domain.
It is then evident that $\varphi(\vec{y},z)$ is minimized for a $z$ such that $G(\vec{y},z)$ is supported (if such a $z$ exists).
}

It is now straightforward to check that the following holds: 
\begin{proposition} If  $G$ and $ \varphi$ are $\Logic_\K$-formulas, $\DB$ is a $\K$-database, and $\alpha$ is an assignment, we have that
\(
    \FOKeval{\nabla_G x.\varphi(x)}{\DB,\alpha} 
     = \min_{a\in D:\FOKeval{G}{\DB,\alpha(a/x)}\neq 0} \FOKeval{\varphi}{\DB,\alpha(a/x)}.
\)
\end{proposition}

Next, we define what rewritable means in our setting:
\begin{definition}\label{def:semiring-rewritable}
    A \emph{semiring rewriting} of $\m\CA_\K(q)$ is an $\Logic_\K$-formula $\varphi_q$ such that for every $\K$-database $\DB$ and assignment $\alpha$, we have that $\m\CA_\K(q, \DB,\alpha) = \varphi_q(\DB,\alpha)$.
    %
    We say that $q$ is \emph{$\Logic_\K$-rewritable} if there exists a semiring rewriting of $q$.
\end{definition}
\begin{example}\label{ex:qpath}  
\suggestion{}{
Consider $\qpath$ and the expression $\epath$ defined in Section \ref{sec:qpath}. 
Observe that $\m\CA_\K(\qpath)$ is $\Logic_\K$-rewritable, since the $\Logic_\K$-formula 
\(
\exists x \nabla_{R(x,y)} y. \,({R(x,y) \times \nabla_{S(y,z)} z .\, {S(y,z)}})
\)
evaluates to \epath, over any $\K$-database $\DB$.
}

\suggestion{}{
When $\K=\B$ (i.e., over the Boolean semiring $\B$), the guarded minimization operators in the latter formula can be eliminated using \eqref{eq:nabla}, leading to a formula whose interpretation matches~\eqref{eq:fuxmannew} and thus coincides with existing rewritings~\cite{DBLP:journals/jcss/FuxmanM07, KoutrisW17}. However, this elimination is not possible for semirings in general.}
\end{example}

\subsection{$\K$-circuits and complexity theory}\label{sec:logic_good}

Our goal is to use $\Logic_\K$ as a logic for rewritings of $\m\CA_\K(q,\Sigma)$. 
Note that,  from Proposition \ref{prop:semiring_to_Boolean}, it follows directly  that any general rewriting result for our logic embeds the rewriting of Koutris and Wijsen for ordinary databases.
We need to establish that $\Logic_\K$ maintains some of the benefits of  $\FO$ mentioned in the beginning of Section \ref{sec:mca}. In particular, the data complexity of $\FO$ is in $\DLOGTIME$-uniform $\aco$ (and hence in \PTIME).
Arguably, a rewriting of $\m\CA_\K(q)$ in $\Logic_\K$ retains the conceptual benefit of query rewriting: instead of computing answers of $q$ for each repair, one may compute the answers of the rewriting directly on the inconsistent database. 
In order to argue about the complexity of $\Logic_\K$, and hence about the suitability of the logic for rewriting, we need to introduce a generalization of $\aco$ to semirings. For this purpose, we define a variant of semi-unbounded fan-in arithmetic $\aco$. We emphasize that, in general, an input to $\m\CA_\K(q,\Sigma)$ is a $\K$-database, where $K$ could be infinite (such as the real numbers), and hence a suitable model of computation has to be able to deal with arithmetics over an arbitrary semiring.
%
%

Next, we give the  definitions concerning circuits that are needed  for our purpose. We refer the reader to the book \cite{Vollmerbook} for a thorough introduction to circuit complexity and to the book \cite{jukna2023tropical} for an exposition of tropical circuits as a tool for studying discrete optimization problems via dynamic programming.
$\K$-circuits are a model of computation for computing semiring-valued functions. 
\begin{definition}\label{def:circuit}
    Let $\K$ be a naturally ordered positive semiring.
    A \emph{$\K$-circuit with min} is a finite simple directed acyclic graph of labeled nodes, also called \emph{gates}, such that
    \begin{itemize}[noitemsep,topsep=0pt]
        \item there are gates labeled \emph{input},  each of which has indegree $0$,
        \item there are gates labeled \emph{constant}, with indegree $0$ and labeled with a $c \in K$,
        \item there are gates labeled \emph{addition}, \emph{multiplication}, \emph{min}, and \suggestion{}{$\supp$},
        \item exactly one gate of outdegree $0$ is additionally labeled \emph{output}.
    \end{itemize}
    Additionally, the input gates are ordered. Note that addition, multiplication, and min gates can have arbitrary in-degree, and the \suggestion{}{$\supp$-gates have in-degree 1}. The \emph{depth} of a circuit $C$ is the length of the longest path from an input gate to an output gate in $C$, while the \emph{size} of $C$ is the number of gates in $C$.

    A circuit $C$ of this kind, with $n$ input gates, computes the function $f_C \colon K^n \to K$ as follows:
    First, the input to the circuit is placed in the input gates.
    Then, in each step, each gate whose predecessor gates all have a value, computes the respective function it is labeled with, using the values of its predecessors as inputs.
    The output of $f_C$ is then the value of the output gate after the computation.
\end{definition}

Each  $\K$-circuit computes a function with a fixed number of arguments; for this reason, we consider families $(C_{n})_{n \in \N}$ of $\K$-circuits, where each circuit $C_{n}$ has exactly $n$ input gates. A  family $\C = (C_{n})_{n \in \N}$ of $\K$-circuits computes the function $f_\C \colon K^* \to K$ defined as 
    \(
        f_\C(\vec{x}) \coloneqq f_{C_{\lvert \vec{x} \rvert}}(\vec{x}).
    \)

\begin{remark}
    A frequent requirement for considering a circuit family $(C_n)_{n \in \N}$ as an algorithm is the existence of an algorithm that given $n$ outputs the circuit $C_n$.
    Circuit families for which such an algorithm exists are called \emph{uniform}.
    The data complexity of $\FO$ is $\DLOGTIME$-uniform $\aco$ 
    \cite{DBLP:journals/jcss/BarringtonIS90},
    and hence there is a $\DLOGTIME$ algorithm that describes $C_n$, given $n$. More formally, the algorithm takes two numbers $(i,j)$ as an input and outputs the type of the $i$th gate of $C_n$, the index of its $j$th predecessor, and, if the gate is an input gate, the index of the input string it corresponds to.
    See \cite{Vollmerbook} for details on circuit uniformity and \cite{BSSbook} for uniformity in the context of real computation.
\end{remark}

\begin{definition}
For a naturally ordered positive semiring $\K$, let $\aco_\K(+,\times, \min,\suggestion{}{\supp})$ consist of all families $(C_n)_{n \in \N}$ of $\K$-circuits as defined in Definition \ref{def:circuit} that are of constant depth and of polynomial size. This means that there is some constant $k\in\N$ and a univariate polynomial function $f\colon \N \to \N$ such that, for each $i\in\N$, the circuit $C_i$ is of depth $k$ and has at most $f(n)$ many gates. We write $\aco_\K(+,\times_2, \min, \suggestion{}{\supp})$ for the restriction of $\aco_\K(+,\times, \min, \suggestion{}{\supp})$ in which all  multiplication gates have in-degree $2$. We write $\faco_\K(+,\times, \min, \supp)$ and $\faco_\K(+,\times_2, \min, \suggestion{}{\supp})$, for the class of functions computed by $\aco_\K(+,\times, \max, \suggestion{}{\supp})$ and $\aco_\K(+,\times_2, \min, \suggestion{}{\supp})$ circuits, respectively.
\end{definition}
Note that, for the Boolean semiring, the classes $\aco_\B(+,\times, \min, \supp)$, $\aco_\B(+,\times_2, \min, \suggestion{}{\supp})$, and $\aco$ coincide, for in this case  $\times$ and $\min$ correspond to $\land$-gates, $+$ corresponds to $\lor$-gates, and \suggestion{}{$\supp$ corresponds to not-gates}. Hence, by a classical result by Immerman relating $\FO$ and $\aco$, the functions in $\DLOGTIME$-uniform $\faco_\B(+,\times_2, \min, \suggestion{}{\supp})$ are precisely those that can be defined in $\Logic_\B$.


The proof of the next result is  in Appendix~\ref{appendix_circuits}.
\begin{restatable}{proposition}{circuitprop}\label{prop:circuitprop}
For every naturally ordered positive semiring $\K$, the data complexity of $\Logic_\K$ is in $\DLOGTIME$-uniform $\faco_\K(+,\times_2, \min,\suggestion{}{\supp})$.
\end{restatable}

If we can argue that functions in $\faco_\K(+,\times_2, \min, \suggestion{}{\supp})$ are in some sense simple, we have grounds for $\Logic_\K$ being a ``good'' logic for rewritings.
Consider an $\aco_\K(+,\times_2, \min, \suggestion{}{\supp})$ circuit family $(C_n)_{n\in\mathbb{N}}$ with depth $k$ and $p(n)$ many gates, for some polynomial function $p$. The computation of an $\aco_\K(+,\times_2, \min, \suggestion{}{\supp})$-circuit differs from a computation of an $\aco_\K(+,\times_2)$ circuit only in the additional minimization \suggestion{}{and $\supp$} gates. In its computation, the circuit $C_n$ needs to evaluate at most $p(n)$ minimization \suggestion{}{and $\supp$} gates. To evaluate a single minimization gate, it suffices to make at most $p(n)$-many comparisons between semiring values to find the smallest input to the gate. \suggestion{}{Similarly, to evaluate a $\supp$-gate, it suffices to make one $a=0$-comparison for a semiring value $a$}.  Hence, in comparison to $\aco_\K(+,\times_2)$ circuits, the evaluation of a $\aco_\K(+,\times_2, \min, \suggestion{}{\supp})$ circuit needs to make additionally polynomially many comparisons between two semiring values. 
Therefore, if $\aco_\K(+,\times_2)$ is computationally favorable and size comparisons between semiring values can be made efficiently, then we have an argument that $\aco_\K(+,\times_2, \min, \suggestion{}{\supp})$ is computationally favorable as well. It is easy to see that $\aco_\K(+,\times_2)$ circuit families compute polynomial functions of constant degree, and hence its computation is in a strong sense polynomial. Thus, functions in $\faco_\K(+,\times_2, \min, \suggestion{}{\supp})$ are in a strong sense polynomial.

\subsection{Acyclicity of the attack graph and semiring rewriting of consistent answers}\label{sec:rewriting}

We will next establish a necessary and sufficient condition for $\m\CA_\K(q,\Sigma)$ to be rewritable in $\Logic_\K$,
when $q$ is a self-join free conjunctive query and $\Sigma$ is a set of key constraints, one for each relation in $q$. 
%
To this effect, we  consider the concepts of an attack and of an attack graph,
two concepts that were introduced by Wijsen \cite{DBLP:conf/pods/Wijsen10,DBLP:journals/tods/Wijsen12} and also used in \cite{KoutrisW17,A-Wijsen24}.
Our main result asserts that the attack graph of $q$ (see Definition~\ref{def:attack_graph}) is acyclic if and only if $\m\CA_\K(q,\Sigma)$ is $\Logic_\K$-rewritable.


\begin{restatable}{theorem}{rewritabilitymainthm}
\label{thm:main}
    Let $\K$ be a naturally ordered positive semiring,
    $q$ be a self-join free conjunctive query, and $\Sigma$ a set of key constraints, one for each relation in $q$. The attack graph of $q$ is acyclic if and only if $\m\CA_\K(q,\Sigma)$ is $\Logic_\K$-rewritable.
\end{restatable}
The proof of the theorem is divided into two parts. We first establish the easier right-to-left direction (Proposition \ref{prop:leftarrow_case}) before building up to prove the more involved left-to-right direction. 

We write $\var(R_i)$ and $\var(q)$ to denote the sets of variables that occur in $R_i$ and $q$, respectively.
We write $q[x]$ to denote the CQ obtained from $q$ by removing the quantifier $\exists x$, if there is one.  
If $R(\vec{y};\vec{z})$ is an atom of $q$, we write $\qminusR$ for the query resulting by removing that atom from $q$, as well as the existential quantifiers binding the variables that appear only in the atom $R(\vec{y};\vec{z})$.
We say that an atom $R(\vec{y};\vec{z})$ \emph{induces} the functional dependency $\key{R} \rightarrow \var(R)$ (or $\var(\vec{y}) \to \var(R)$). 
If $q$ is a \sjfCQ, 
we define $\Sigma(q)$ to be the set of functional dependencies \emph{induced} by the atoms in $q$. Formally,
\[
\Sigma(q) \dfn \{ \key{R} \to \var(R) \mid R \in q \}.
\]
The \emph{closure} of a set $X$  of variables  with respect to a set $\Lambda$ of functional dependencies  contains all the variables $y \in \var(q)$ such that $\Lambda \models X \to y$.
Here, $\Lambda \models  X \to y$ means that the set of functional dependencies in $\Lambda$ semantically entails the functional dependency $X \to y$. Syntactically, this entailment can be derived  using Armstrong's axioms for functional dependencies \cite{armstrong74}.
For $R(\vec{y};\vec{z}) \in q$, 
we write \closureVars for the \emph{closure of $\var(\vec{y})$ with respect to $\Sigma(\qminusR)$}, that is 
\[
\closureVars \dfn \{ x \in \var(q) \mid \Sigma(\qminusR\})\models \var(\vec{y})\to x \}.
\]




\begin{definition}\label{def:attack}
    Let $q$ be a \sjfCQ. An atom $R(\vec{y}; \vec{z})$ of $q$ \emph{attacks a variable $x$} bound in $q$ if there exists a non-empty sequence of variables $x_1, \ldots, x_n$ bound in $q$, such that:
    \begin{enumerate}
        \item $x_1 \in \var(\vec{z})$ (i.e., is a non-key variable of $R$), and $x_n=x$;
        \item for every $i <n$, we have that $x_i, x_{i+1}$ occur together in some atom of $q$; and
        \item 
        for every $i\leq n$, we have that $x_i \not\in \closureVars$.
    \end{enumerate}
        A variable $x$ in $q$ is \emph{unattacked in $q$} if no atom of $q$ attacks it.
\end{definition}

\begin{definition}[Attack graph]\label{def:attack_graph}
    The \emph{attack graph} of $q$ is the graph defined as follows:
   (i)  the vertices are the atoms of $q$; (ii) there is a
    directed edge from $R(\vec{y};\vec{z})$ to a different atom $R'(\vec{y'};\vec{z'})$ if $R(\vec{y}; \vec{z})$ attacks a variable in $\var(\vec{y'})$, 
    i.e., if $R(\vec{y};\vec{z})$ attacks a key variable of $R'$ that is bound in $q$.
    %
\end{definition}




Our definition of an attack graph is phrased slightly differently but remains equivalent to the definitions found in~\cite{KoutrisW17,A-Wijsen24}. Specifically, note that if an atom attacks a bound variable in some other atom~$R$, then it necessarily also attacks a bound variable that occurs in $\key{R}$.

%

The next proposition establishes the right-to-left direction of Theorem \ref{thm:main}.
%

\begin{proposition}\label{prop:leftarrow_case}
Let $\K$ be a naturally ordered positive semiring, $q$ be a self-join free conjunctive query, and $\Sigma$ a set of key constraints, one for each relation in $q$. If 
$\m\CA_\K(q,\Sigma)$ is $\Logic_\K$-rewritable, then 
the attack graph of $q$ is acyclic.
\end{proposition}
\begin{proof}
Assume that $\varphi$ is 
 $\Logic_\K$-rewriting of $\m\CA_\K(q,\Sigma)$. 
  Let $\psi$ be the corresponding $\FO$-formula obtained from Proposition \ref{prop:semiring_to_Boolean}. It is easy to see that $\psi$ is an $\FO$-rewriting of $\CONS(q,\Sigma)$ on ordinary databases. 
  From the results in \cite{KoutrisW17}, it follows that the attack graph of $q$ 
  is acyclic.
  \end{proof}




Towards proving the converse direction of Theorem~\ref{thm:main}, we start with the next lemma, whose proof can be found in Appendix~\ref{appendix_rewriting}:


\begin{restatable}{lemma}{jefslemmatwo}\label{lemma:jefs_lemma2}
If $q$ is a \sjfCQ with an acyclic attack graph and $x$ is an unattacked variable, then $q[x]$ has an acyclic attack graph. Moreover, the attack graph of $q[x]$ is a subgraph of the attack graph of $q$.
\end{restatable}

\begin{lemma}\label{lemma:jef}
    Let $q$ be a \sjfCQ and $\ICs$ be a set of key constraints, one key per relation.
    Let $\DB$ be a $\K$-database and $\gamma$ an assignment such that $\DB',\gamma \models_\K q$ for every repair $\DB'$ of $\DB$ (with respect to $\ICs$), and let $x$ be an unattacked variable that is bound in $q$. Then, there is an element $c \in D$ such that $\DB',\gamma(c/x) \models_\K q[x]$, for every repair $\DB'$ of $\DB$.\looseness=-1
\end{lemma}
\begin{proof}
    Let $\DB$ and $\gamma$ be as described in the statement of the lemma.
    For every repair $\DB'$ of $\DB$, let $\rho(\DB')= \{ c \in D \mid \DB',\gamma(c/x) \models_{\K} q[x] \}$. This set is always nonempty, since by assumption, $\DB', \gamma \models_\K q$ for every repair $\DB'$ of $\DB$. We will prove the following stronger formulation of the lemma:
    \begin{claim}\label{claim:stronger}
    There is a repair $\RP^*$ of $\DB$ such that $\rho(\RP^*) = \bigcap_{\DB' \text{ repair of } \DB} \rho(\DB') \neq \emptyset$.        
    \end{claim} 
    Let $\RP^*$ be a repair of $\DB$ for which $\rho(\RP^*)$ is minimal according to subset inclusion. We will show that $\rho(\RP^*) \subseteq \rho(\RPS)$, for every repair $\RPS$ of $\DB$.
    To this end, fix an arbitrary repair $\RPS$ of $\DB$, and let $\RP$ be a repair such that $\rho(\RP) = \rho(\RP^*)$ and, in addition, there is no other repair $\mathfrak T$ of $\DB$ such that $\rho(\RP^*)=\rho(\mathfrak T)$ and $\mathfrak T \cap \RPS \supsetneq \RP \cap \RPS$. 
    %
    We will show that $\rho(\RP) \subseteq \rho(\RPS)$. Take $a \in \rho(\RP)$. 
    We want to show that $a \in \rho(\RPS)$. Towards a contradiction, assume $a \not\in \rho(\RPS)$.
    Hence, $\RP, \gamma(a/x) \models_{\K} q[x]$ and $\RPS, \gamma(a/x) \not\models_{\K} q[x]$. Let $\alpha$ be a valuation such that $\RP, \alpha \models_\K \hat{q}$ and agrees with $\gamma(a/x)$ with respect to the free variables of $q[x]$.
    Thus, there is a fact $A \in \alpha(\hat{q})$ such that $A \in \RP$ and $A \not\in \RPS$. Since $\RPS$ is a repair of $\DB$ it contains a fact $A'$ that is key-equal to $A$.
    Now, consider $\RP' = (\RP \setminus \{A\}) \cup \{A'\}$. Note that, $\RP'$ is a repair of $\DB$, for we are considering only key constraints with one key per relation and sjfCQs. Since $\RP' \cap \RPS \supsetneq \RP \cap \RPS$, it follows from the choice of $\RP$ that $\rho(\RP') \neq \rho(\RP)$.
    We will show the following claim, which will then lead to a contradiction.
    \begin{claim}\label{claim:rho_containment}
        $\rho(\RP') \subsetneq \rho(\RP)$
    \end{claim}
    \begin{proof}
        Since $\rho(\RP') \neq \rho(\RP)$, we only need to show containment.
        Pick an arbitrary $b \in \rho(\RP')$, and notice that $\RP', \gamma(b/x) \models_{\K} q[x]$.
        Let $\beta$ be a valuation such that $\RP', \beta \models_\K  \hat{q}$ and agrees with $\gamma(b/x)$ with respect to the free variables of $q[x]$.
        The claim follows if $\RP, \beta \models_{\K} q[x]$, and thus assume $\RP', \beta \not\models_\K  q[x]$. Since $\RP$ and $\RP'$ contain the same facts with the exception of $A$ and $A'$, then necessarily $A' \in \beta(\hat{q})$.
        Let $R(\vec{y};\vec{z})$ be the atom in $q$ with the same relation name as $A$ and $A'$, and let $\delta$ be the valuation for the variables in $q$ defined as follows:
        \[
        \delta(v) = \begin{cases}
            \alpha(v) & \text{ if } R(\vec{y};\vec{z}) \text{ attacks } v \text{ in } q \\
            \beta(v) & \text{ otherwise}
        \end{cases}
        \]
        Since $x$ is unattacked, it is not attacked by $R$, and thus $\delta(x)=\beta(x) = b$. Hence, to prove that $b\in\rho(\RP)$ it suffices to show that $\RP, \delta \models_\K \hat{q}$.
        Recall that $A$ and $A'$ are key-equal facts of relation name $R$ such that $A\in \alpha(\hat{q})$ and $A'\ \in \beta(\hat{q})$. Hence $\alpha(\vec{y})= \beta(\vec{y})$ for the key variables $\vec{y}$ of $R$.

        Let $\Sigma(\qminusR)$, and let $\closureVars$. 
        Notice that, if we define 
        \begin{align*}
            Y_0 &\eqdef \var({\vec{y}}) \\ 
            Y_{n+1} &\eqdef \{ w \in \Var \mid \mbox{ there is } (V \to W)\in \Sigma(\qminusR) \mbox{ such that } V \subseteq Y_n \mbox{ and } w \in W \}
        \end{align*} 
        then $\closureVars = \bigcup_n Y_n$.
        
        Clearly, $\var({\vec{y}}) \subseteq \closureVars$.
        We now show that $\alpha$ and $\beta$ agree on all variables in $\closureVars$. More precisely, we show that $\alpha(w) = \beta(w)$ for every $w \in \closureVars = \bigcup_n Y_n$ by induction on $n$.
        For $Y_0$ this follows, for we already established that $\alpha(\vec{y})= \beta(\vec{y})$ for the key variables of $R$.
        Now assume $\alpha(Y_n)= \beta(Y_n)$, we want to show that $\alpha(Y_{n+1})= \beta(Y_{n+1})$.
        Let $w \in Y_{n+1}\setminus Y_{n}$, then there is a functional dependency $V \to W \in \Sigma(\qminusR)$ such that $V \subseteq Y_n$ and $w\in W$. 
        This dependency is of the form $\vec{y'} \to \vec{z'}$ for some atom $S(\vec{y'}; \vec{z'}) \neq R(\vec{y}; \vec{z})$, where $\vec{y'} \subseteq V \subseteq Y_n$ and $\vec{z'} \subseteq W$. By inductive hypothesis, $\alpha(\vec{y'})= \beta(\vec{y'})$. Furthermore, $\RP$ and $\RP'$ agree on every atom with the exception of $A$ and $A'$ which are the facts corresponding to the relation $R$, thus $\RP$ and $\RP'$ agree on $S(\vec{y'}; \vec{z'})$ and this shows that $\alpha$ and $\beta$ agree on all variables in $Y_{n+1}$.
        It follows by the induction principle that $\alpha$ and $\beta$ agree on all variables in $\closureVars$.

        Notice now that for every non-key variable $z$ of $R$, if $z \not\in \closureVars$ then $R(\vec{y}; \vec{z})$ attacks~$z$, since in that case the sequence of length one consisting only of $z$ witnesses the attack, and consequently $\delta(z)=\alpha(z)$. Furthermore, since $\alpha$ and $\beta$ agree on variables in $\closureVars$, $\delta$ maps the atom $R(\vec{y}; \vec{z})$ to the fact $R(\alpha(\vec{y}); \alpha(\vec{z})) = A$, which belongs to $\RP$.

        Let $S(\vec{y'}; \vec{z'})$ be an arbitrary atom of $\qminusR$.
        Since $\RP,\alpha \models_\K \hat{q}$ and $\RP',\beta \models_\K \hat{q}$ there are facts $B\in \RP$ and $B'\in \RP'$ such that $\alpha$ and $\beta$ map $S(\vec{y'}; \vec{z'})$ to $B$ and $B'$, respectively.
        Since $\RP$ and $\RP'$ contain the same facts with the exception of $A$ and $A'$, we have that $B, B' \in \RP$.
        Hence, it suffices to show that $\delta$ maps $S(\vec{y'}; \vec{z'})$ to either $B$ or $B'$. 
        
        If all variables in $S(\vec{y'}; \vec{z'})$ are attacked (not attacked, resp.) by $R(\vec{y}; \vec{z})$ then, by definition, $\delta$ coincides with $\alpha$ ($\beta$, resp.) with respect to the variables in $S$ and hence maps $S(\vec{y'}; \vec{z'})$ to $B$ ($B'$, resp.).

        If we are not in the above case there are both variables that are attacked and not attacked by $R(\vec{y}; \vec{z})$. Let $x_1$ and $x_2$ be arbitrary variables that occur in $S(\vec{y'}; \vec{z'})$ such that $R(\vec{y}; \vec{z})$ attacks $x_1$ but does not attack $x_2$. Towards a contradiction suppose $x_2 \not\in \closureVars$. Since $R$ attacks $x_1$, there is a sequence $v_1,\ldots, v_k$ of variables starting in a non-key variable of $R$ and ending in $x_1$ witnessing the attack (see Definition \ref{def:attack}). In particular is $v_k=x_1$. Since $x_1$ and $x_2$ belong to the same atom $S$ in $q$ and $x_2 \not\in \closureVars$, we can extend the sequence by setting $v_{k+1}=x_2$ thus obtaining a witness for the attack from $R$ to $x_2$, which results in a contradiction since $x_2$ is not attacked by our assumption.
        Hence, $x_2 \in \closureVars$ and thus $\alpha(x_2) = \beta(x_2)=\delta(x_2)$. Also, $\delta(x_1)=\alpha(x_1)$ by construction of $\delta$. Since, $x_1$ and $x_2$ were arbitrary, $\delta$ maps $S(\vec{y'}; \vec{z'})$ to $B$, which is in $\RP$. 
        
        We have shown that $\delta$ maps $S(\vec{y'}; \vec{z'})$ to either $B$ or $B'$, and hence to a fact in $\RP$. Therefore, $\delta$ maps every atom in $q$ to a fact in $\RP$ and thus $\RP, \delta \models_\K \hat{q}$, which then finishes the proof of Claim~\ref{claim:rho_containment}.
    \end{proof}
    Since we assumed that $\RP$ is a repair of $\DB$ such that $\rho(\RP)$ is subset-minimal, Claim~\ref{claim:stronger} now follows by a contradiction given by Claim~\ref{claim:rho_containment}, which concludes the proof of the lemma.
\end{proof}



The proof of the following lemma can be found in Appendix~\ref{appendix_rewriting}.
\begin{restatable}{lemma}{lemmaKDBs}
\label{lemma:lemma1_for_Kdbs}
    Let $q$ be a \sjfCQ and $\Sigma$ a set of key constraints, one per relation, with an acyclic attack graph. Let $\alpha:\Var\to A$ be an assignment and  $x$ be an unattacked variable in $q$. Then, for every $\K$-database $\DB$,
    \[
     \m\CA_\K(q, \DB, \alpha) = \sum_{c \in D} \m\CA_\K(q[x], \DB, \alpha(c/x)).
    \]
\end{restatable}

We are now ready to give the proof of our main result. 
\rewritabilitymainthm*


\begin{proof}
    The right-to-left direction follows from Proposition~\ref{prop:leftarrow_case}.
    We prove the left-to-right direction by induction on $|q|$.
    We will simultaneously define the rewriting and prove its correctness. 
    We first prove the correctness of the rewriting for $\K$-databases $\DB$ and assignments $\alpha$ such that $\DB',\alpha \models_\K q$, for every repair $\DB'$ of $\DB$, then at the end of the proof argue that the rewriting is correct also for the remaining case.
    Let $q$ be a self-join-free conjunctive query whose attack graph is acyclic. \looseness=-1
    %
    %
    For the base case, assume $|q|= 1$, that is, $q=\exists \vec{x} R(\vec{y}; \vec{z})$ where $\vec{x} \subseteq \vec{y}\cup\vec{z}$.
    Since $q$ is acyclic, there exists an unattacked atom in $q$, which in this case is $R(\vec{y}; \vec{z})$. Set $\vec{y}_{\vec{x}}$ (resp.\ $\vec{z}_{\vec{x}}$) to be the list of variables that are in $\vec{y}$ (resp.\ in $\vec{z}$) and $\vec{x}$.
    By Lemma~\ref{lemma:lemma1_for_Kdbs}, we get
    \[
     \m\CA_\K(q, \DB, \alpha) = \sum_{\vec{a} \in D^{\vert \vec{y}_{\vec{x}}  \rvert } } \m\CA_\K(q[\vec{y}_{\vec{x}}], \DB, \alpha'),  \text{ where $\alpha' \dfn \alpha(\vec{a}/\vec{y}_{\vec{x}})$ }.
    \]
    
    \noindent By definition of $\m\CA$ and since there is exactly one key-equal fact to $R(\alpha'(\vec{y}); \vec{z})$ in every repair $\DB'$ of $\DB$,
    \[
    \m\CA_\K(q[\vec{y}_{\vec{x}}], \DB, \alpha') = \min_{\DB' \in \Rep(\DB, \Sigma)} q[\vec{y}_{\vec{x}}](\DB', \alpha') =  \min_{\vec{b} \in D^{|\vec{z}_{\vec{x}}|} : R^{\DB}(\beta(\vec{y}) ;\beta(\vec{z})) \neq 0} R^{\DB}(\beta(\vec{y});\beta(\vec{z})), \tag{\theequation}\label{eq:proof_min_base_case}
    \]
    where $\beta \dfn \alpha'(\vec{b}/\vec{z}_{\vec{x}})$. 
    This yields the following $\Logic_\K$-rewriting $\varphi_q$ for the base case:
    \[
    \varphi_q\dfn \exists \vec{y}_x \nabla_{R(\vec{y};\vec{z})} \vec{z}_{\vec{x}}.R(\vec{y};\vec{z}).
    \]
    If there is a repair $\DB'$ of $\DB$ for which $\DB',\alpha\not\models_\K q$, since the semiring is positive, the last minimization in (\ref{eq:proof_min_base_case}) is over the empty set for every $\vec{a} \in D^{|\vec{y}_{\vec{x}}|}$, thus $\varphi_q(\DB,\alpha) = 0 = \m\CA_k(q, \DB, \alpha)$.
    For the inductive step, suppose that $|q| =n > 1$, and that the claim holds for every acyclic sjfCQ of size at most $n-1$. 
    Let $R(\vec{y}; \vec{z})$ be an unattacked atom in $q$, which exists since $q$ is acyclic. 
    As done in the base case, we define $\vec{y}_{\vec{x}}$ and $\vec{z}_{\vec{x}}$, we use Lemma~\ref{lemma:lemma1_for_Kdbs}, and focus on rewriting the expression    
    %
    \( 
    \m\CA_\K(q[\vec{y}_{\vec{x}}], \DB, \alpha'),
    \)
    where $\alpha' \dfn \alpha(\vec{a}/\vec{y}_{\vec{x}})$.
    We want to show that:
    \begin{align}
    \min_{\DB' \in \Rep(\DB,\Sigma)} q[\vec{y}_{\vec{x}}] (\DB', \alpha') 
    &= \min_{\vec{b} \in D^{|\vec{z}_{\vec{x}}|}: R^{\DB}(\beta(\vec{y});\beta(\vec{z}))\neq 0} \big( R^{\DB}(\beta(\vec{y});\beta(\vec{z}))  \times \min_{\DB' \in \Rep(\DB,\Sigma)} \qyzminusR(\DB', \beta) \big) \nonumber \\
    &= \min_{\vec{b} \in D^{|\vec{z}_{\vec{x}}|}: R^{\DB}(\beta(\vec{y});\beta(\vec{z}))\neq 0} \big( R^{\DB}(\beta(\vec{y});\beta(\vec{z}))  \times \m\CA_\K (\qyzminusR, \DB, \beta) \big) \label{eq:proof_rightarrowcase}
    \end{align}
    where $\qyzminusR = q[\vec{y}_{\vec{x}}][\vec{z}_{\vec{x}}]\setminus R(\vec{y};\vec{z})$, and $\beta \dfn \alpha'(\vec{b}/\vec{z}_{\vec{x}})$.
    It is easy to see that: 
    \[
    \min_{\DB' \in \Rep(\DB,\Sigma)} q[\vec{y}_{\vec{x}}] (\DB', \alpha') 
    \geq_\K \min_{\vec{b} \in D^{|\vec{z}_{\vec{x}}|}:  R^{\DB}(\beta(\vec{y});\beta(\vec{z}))\neq 0} \big( R^{\DB}(\beta(\vec{y});\beta(\vec{z}))  \times \min_{\DB' \in \Rep(\DB,\Sigma)} \qyzminusR(\DB', \beta) \big).
    \]

    To show the other direction, assume that $\m\CA_\K(q[\vec{y}_{\vec{x}}], \DB, \alpha')>0$, and let $\DB_{\min}$ be a repair 
    such that $q[\vec{y}_{\vec{x}}] (\DB_{\min}, \alpha')$ is minimum over all repairs of $\DB$. 
    We will prove that, for every repair $\DB^*$ of $\DB$ and for every $\vec{b} \in D^{|\vec{z}_{\vec{x}}|}$ such that $R^{\DB}(\beta(\vec{y}); \beta(\vec{z})) \neq 0$:
    \begin{align} \label{eq:prop4.16}
    q[\vec{y}_{\vec{x}}] (\DB_{\min}, \alpha') 
    \leq_\K R^{\DB}(\beta(\vec{y});\beta(\vec{z})) \times \qyzminusR(\DB^*, \beta).
    \end{align}

    It follows, as in the proof of Lemma~\ref{lemma:lemma1_for_Kdbs}, that 
    \(
    q[\vec{y}_{\vec{x}}] (\DB_{\min}, \alpha') 
    \leq_\K q[\vec{y}_{\vec{x}}](\DB^*, \alpha').
    \)
    Since there is exactly one key-equal fact (with a non-zero annotation) with relation name $R$ per repair, in particular in $\DB^*$, the latter can be rewritten as $R^{\DB}(\beta(\vec{y});\beta(\vec{z})) \times \qyzminusR(\DB^*, \beta)$. This proves \eqref{eq:prop4.16}, which then gives us the desired expression for $\m\CA_\K(q[\vec{y}_{\vec{x}}], \DB, \alpha')$.  

    On the other hand, since $|\qyzminusR|<n$, it follows from the induction hypothesis that
    \(
    \m\CA_\K (\qyzminusR, \Sigma) 
    \) 
    admits an $\Logic_\K$-rewriting $\varphi_{\widetilde{q}}$, hence we obtain the following $\Logic_\K$-rewriting for $\m\CA_\K(q, \Sigma)$:
    \[
    \varphi_q \dfn \exists \vec{y}_{\vec{x}} \nabla_{R(\vec{y};\vec{z})} \vec{z}_{\vec{x}}.\varphi_{\widetilde{q}}.
    \]
    
    %
    If there exists a repair $\DB'$ such that $\DB',\alpha \not\models_\K q$, then 
    $\m\CA(q, \DB,\alpha) = 0$. From this and the hypothesis that the semiring $\K$ is positive, it follows that when we inspect \eqref{eq:proof_rightarrowcase} for every $\vec{a} \in D^{|\vec{y}_{\vec{x}|}}$, either the set $\{ \vec{b} \in D^{|\vec{z}_{\vec{x}}|} : R^{\DB}(\beta(\vec{z}),\beta(\vec{z}))\neq 0 \}$ is empty and thus the minimization is done over an empty set, or the minimization of $\qyzminusR(\RP, \alpha)$ over all repairs $\RP$ is $0$, hence $\varphi_q(\DB,\alpha) = 0$.
    This completes the proof of the theorem. \looseness=-1
\end{proof}

The preceding theorem generalizes the Koutris-Wijsen rewritability result in \cite{KoutrisW17}. Furthermore, it gives new rewritability results for the bag, the tropical, the Viterbi, and   the fuzzy semiring.




\section{\suggestion{}{From rewritability to non-approximability}}

Since the early days of computational complexity, there has been an extensive study of the \emph{approximation} properties of optimization problems, i.e., whether or not there are ``good'' algorithms that approximate the optimum (maximum or minimum) value of some objective function defined on a space of feasible solutions.
Most of the work on the approximation properties of optimization problems has focused on optimization problems
whose underlying decision problem is \NP-complete.
This study has shown that such optimization problems may have very different approximation properties, ranging from polynomial-time approximation schemes (e.g., {\sc Knapsack})  to  polynomial-time constant-approximation algorithms (e.g., {\sc Min Vertex Cover}) to polynomial-time logarithmic-approximation algorithms (e.g., {\sc Min Set Cover}), or even worse approximation properties (e.g., {\sc Max Clique}) 
- see \cite{DBLP:books/fm/GareyJ79,DBLP:books/daglib/0072413}. There has also been work on the approximation properties of optimization problems whose underlying decision problem is 
in some lower complexity class, such as \Log and \NL, where \Log and \NL denote,  respectively, deterministic log-space and non-deterministic log-space \cite{DBLP:journals/mst/Tantau07}. 

Let ${\K}=(K,+,\times,0,1)$ be a naturally ordered positive semiring.  Every closed query $q$ and  every  set $\Sigma$ of  constraints give rise to the optimization problem $\m\CA_\K(q,\Sigma)$: given a $\K$-database $\DB$, compute
$\m\CA_\K(q,\Sigma, \DB)$. The decision problem underlying $\m\CA_\K(q,\Sigma)$ asks: given a $\K$-database $\DB$ and an element $k\in \K$, is
$\m\CA_\K(q,\Sigma,\DB)\leq_\K k$? 
In the case of the Boolean semiring $\mathbb B$, both the optimization problem $\m\CA_{\mathbb B}(q,\Sigma)$ and its underlying decision problem coincide with the decision problem $\CONS(q,\Sigma)$. Furthermore, as mentioned earlier, it has been shown in \cite{KoutrisW17,DBLP:journals/mst/KoutrisW21} that if $q$ is a self-join free conjunctive query with one key constraint per relation,
then $\CONS(q,\Sigma)$ exhibits a trichotomy, which we now spell out in detail: 
(i) if the attack graph of 
$q$ is acyclic, then $\CONS(q,\Sigma)$ is $\FO$-rewritable; 
(ii) if the attack graph of $q$ contains a \emph{weak} cycle but no \emph{strong} cycle, then 
$\CONS(q,\Sigma)$ is \Log-complete, hence it is in \PTIME but it is  not $\FO$-rewritable; 
(iii) if the attack graph of $q$ contains a strong cycle, then $\CONS(q,\Sigma)$ is \coNP-complete (the precise definitions of a weak cycle  and of a strong cycle in the attack graph can be found in  \cite{KoutrisW17}).

\newcommand{\opt}{\mathsf{opt}}

Let $q$ be a closed self-join free conjunctive query and let $\Sigma$ be a set of key constraints, one for each relation of $q$.
In what follows, we will leverage the above trichotomy result for  $\CONS(q,\Sigma)$ to study the approximation properties of computing 
$\m\CA_{\mathbb N}(q,\Sigma)$, where ${\mathbb N}=(N,+,\times, 0,1)$ is the bag semiring.  
Let $\varepsilon \geq 1$ be a fixed constant and consider a minimization problem  $\mathcal Q$ in which the objective function takes positive integers as values.
An $\varepsilon$-approximation algorithm for $\mathcal Q$ is an algorithm that, given an input to $\mathcal Q$,  returns the value $A$ of the objective function on some feasible solution so that  $A / \opt \leq \varepsilon$, where  $\opt$ is the value of the objective function on the given input  (note that $A/\opt \geq 1$ because $\mathcal Q$ is a minimization problem).
Since $\m\CA_{\mathbb N}(q,\Sigma)$ may take the value $0$ on an input $\mathbb N$-database $\DB$, we will consider the minimization problem: given an $\mathbb N$-database $\DB$, compute $\m\CA_{\mathbb N}(q,\Sigma,\DB)+1$.

\newcommand{\approximate}[3]{\mathsf{APPROX}({#1}, {#2},{#3})}
\newcommand{\naturals}{\mathbb{N}}
\newcommand{\mydb}{\mathbf{db}}
\newcommand{\myrep}{\mathbf{r}}
\newcommand{\kdb}{\mathcal{D}}
\newcommand{\krep}{\mathcal{R}}

\newcommand{\certainty}[1]{\mathsf{CERTAINTY}({#1})}

Assume that $q$ is a closed self-join free conjunctive query and $\Sigma$ is a set of key constraints, one for each relation of $q$.  For every $\varepsilon \geq 1$, let
$\approximate{q}{\Sigma}{\varepsilon}$ be the following function problem:

\begin{description}
\item[PROBLEM $\approximate{q}{\Sigma}{\varepsilon}$]
\item[INPUT:] An arbitrary $\naturals$-database $\DB$.
\item[OUTPUT:] The value
$q(\DB')$ of $q$ on some repair
 $\DB'$ of $\DB$ such that $q(\DB')/ \opt \leq \varepsilon$,
where $\opt=\m\CA_\N(q,\Sigma,\DB)+1$.
\end{description}

Clearly, the inequality $q(\DB') / \opt \leq \varepsilon$ encapsulates a relative approximation  guarantee.  For example, if we take $\varepsilon= 1.5$, then $\approximate{q}{\Sigma}{1.5}$
asks for the value $q(\DB')$ of $q$ on some repair $\DB'$ of $\DB$ such that $q(\DB')\leq 1.5 \cdot (\m\CA_\N(q,\Sigma,\DB)+1)$.

\begin{proposition}\label{lem:reduction}
Let $q$ be a closed self-join-free conjunctive query and let $\Sigma$ be a set of key constraints, one for each relation of $q$.
For every $\varepsilon\geq 1$, there is a first-order reduction from 
the complement of {\rm $\cons(q,\Sigma)$}  to $\approximate{q}{\Sigma}{\varepsilon}$.
\end{proposition}
\begin{proof} Fix  $q$, $\Sigma$, and $\varepsilon\geq 1$, as in the hypothesis of this proposition.
Let $M$ be a fixed natural number such that $M> \varepsilon  $ 
and pick a relation symbol $R$ occurring in $q$.

Given a standard database $\DB_0$ (i.e., a $\mathbb B$-database) that is an input to $\CONS(q,\Sigma)$, we construct an $\mathbb N$-database $\DB$ as follows:

\begin{itemize}
\item Annotate every $R$-fact of $\DB_0$ with $M$, that is,
if $t \in \suppo(R^{\DB_0})$, then $R^\DB(t)=M$; otherwise
$R^\DB(t) = 0$.
\item Annotate all other facts in $\DB_0$ with~$1$, that is, for every relation symbol $S$ in $q$ that is different from $R$, if $s\in \suppo(S^{\DB_0})$, then $S^\DB(s)=1$; otherwise $S^\DB(s)=0$.
\end{itemize}
Since $M$ is fixed, the $\mathbb N$-database $\DB$ can be constructed from $\DB_0$ in first-order logic.
Furthermore, it is clear that there is a one-to-one correspondence between the repairs of
$\DB_0$ and the repairs of $\DB$.

We now claim that $\cons(q,\Sigma, \DB_0)$ is false   if and only if $\approximate{q}{\Sigma}{\varepsilon}$ returns $0=q(\DB')$ for some repair $\DB'$ of $\DB$, where, of course, the evaluation of $q(\DB')$ is over the bag semiring $\mathbb N$.

Assume first that $\cons(q,\Sigma, \DB_0)$ is false, which means that  there is a repair $\DB^*_0$ of $\DB_0$ that falsifies $q$.
It is easily verified that there is a (unique) repair $\DB^*$ of $\DB$ such that $\suppo({\DB^*})=\DB_0^*$.
Furthermore, $q(\DB^*)=0$, hence 
$\m\CA_{\mathbb N}(q,\Sigma,\DB)=0$ and $\opt =1$.
Now take any repair $\DB'$ of $\DB$ such that  $q(\DB')\neq 0$.
By our construction, $q(\DB')\geq M> \varepsilon$, hence $q(\DB') / \opt \geq M >\varepsilon $, which implies that $q(\DB')\not = 0$ is not a valid output for $\approximate{q}{\Sigma}{\varepsilon}$.
It follows that $\approximate{q}{\Sigma}{\varepsilon}$ must return $0=q(\DB')$ for some repair $\DB'$ of $\DB$.

Conversely, assume that $\approximate{q}{\Sigma}{\varepsilon}$ returns $0=q(\DB')$ for some  repair $\DB'$ of $\DB$. By our construction, it must be the case that $\suppo(\DB')$ is a repair of $\DB_0$ that falsifies~$q$, hence $\cons(q,\Sigma, \DB_0)$ is false.
This concludes the proof of Proposition~\ref{lem:reduction}.
\end{proof}

\begin{corollary} \label{inapp-cor}
Let $q$ be a closed self-join-free conjunctive query and let $\Sigma$ be a set of key constraints, one for each relation of $q$. Then the following statements are true.
\begin{itemize}
\item If the attack graph of $q$ contains a (weak or strong) cycle, then for every $\varepsilon \geq 1$, the problem $\approximate{q}{\Sigma}{\varepsilon}$ is $\Log$-hard under first-order reductions; and
\item if the attack graph of $q$ contains a strong cycle, then for every $\varepsilon\geq 1$, the problem $\approximate{q}{\Sigma}{\varepsilon}$ is $\NP$-hard under first-order reductions.
\end{itemize}
\end{corollary}
\begin{proof}
As discussed earlier, it is shown in \cite{KoutrisW17}  that the complement of $\cons(q,\Sigma)$ is $\Log$-hard if the attack graph of $q$ contains a cycle, and $\NP$-hard if the attack graph of $q$ contains a strong cycle.
The desired conclusions then follow from Proposition~\ref{lem:reduction}.
\end{proof}

We illustrate Corollary \ref{inapp-cor} with two examples.
\begin{itemize}
\item Let $\qcycle$ be the query $\exists x\exists y(R(x;y)\land S(y;x))$. It is known that the attack graph
of $\qcycle$ contains a weak cycle but not a strong cycle. Thus, for every $\varepsilon\geq 1$, we have that $\approximate{\qcycle}{\Sigma}{\varepsilon}$ is
$\Log$-hard under first-order reductions.
\item Let $\qsink$ be the query $\exists x\exists y\exists z (R(x;z)\land S(y;z))$. It is known that the attack graph
of $\qsink$ contains  a strong cycle. Thus, for every $\rho\geq 1$ and every $\alpha\geq 0$, we have that $\approximate{\qsink}{\Sigma}{\varepsilon}$ is
$\NP$-hard under first-order reductions.
\end{itemize}
\section{\suggestion{}{Directions for Future Research}} \label{sec:concl}

In this paper, we initiated a study of consistent query answering for databases over a naturally ordered positive semiring. The results obtained  suggest several research directions in this area.


First, an interesting open question is to extend our complexity study of $\m\CA_\K(q,\Sigma)$ from rewritability  in $\Logic_\K$ to computability in polynomial time. Consider a self-join free conjunctive query whose attack graph contains a weak cycle, but not a \emph{strong} cycle.  
In the case of standard databases, it was  shown in \cite{KoutrisW17} 
that the consistent answers of such queries are polynomial-time computable, but not \FO-rewritable. Does this result extend and how does it extend to  the consistent answers of such queries for $\K$-databases, where $\K$ is a naturally ordered positive semiring?  Here, among others, we have the problem of how to formulate the ``right'' notion of polynomial-time computability over semirings.
For standard databases, two polynomial-time approaches are known for computing consistent answers to self-join-free conjunctive queries with a cyclic attack graph that has no strong cycles: through a rewriting in a variant of Datalog~\cite{DBLP:journals/mst/KoutrisW21}, or via a more general algorithm developed by Figueira et al.~\cite{DBLP:conf/icdt/FigueiraPSS23}, which can be formulated in some fixpoint logic.
It is an open question whether these approaches can be adapted beyond the Boolean semiring to arbitrary naturally ordered positive semirings.
In this respect, it is also significant that the algorithm of Figueira et al.\ is not restricted to the self-join-free case but can also handle some, though not all, conjunctive queries with self-joins. Their algorithm could provide a pathway for extending the results in the current paper to conjunctive queries with self-joins.
Nevertheless, one should be aware that consistent query answering to conjunctive queries with self-joins is a notorious open problem. For example, for conjunctive queries $q$ with self-joins and primary keys $\Sigma$, the complexity of $\m\CA_\B(q,\Sigma)$ is understood for queries with at most two atoms~\cite{DBLP:journals/pacmmod/PadmanabhaSS24}, but is largely open for queries with three or more atoms.

Second, we initiated an investigation into approximating consistent query answers when the computation of exact results is intractable.
In particular, we showed that if $q$ is a  self-join free conjunctive query whose attack graph contains a strong cycle, then the consistent answers on bag databases (i.e., $\m\CA_\N(q,\Sigma)$) are not approximable in polynomial time, unless $\PTIME=\NP$. How does this result extend to  naturally ordered positive semirings other than the bag semiring?

Third, in a recent paper~\cite{A-Wijsen24}, consistent query answering is studied for primary keys and numerical queries that return a single number obtained by aggregating (e.g., by means of SUM or AVG) the results returned by a self-join-free conjunctive query $q(r)$, where the free variable~$r$ is numerical and ranges over~$\mathbb N$. 
The range semantics established in~\cite{DBLP:journals/tcs/ArenasBCHRS03}  requires computing the greatest lower bound (glb) and the least upper bound (lub) on the answers to the numerical query over all repairs.   
The problem of finding the glb for SUM queries can be restated as a special case of $\m\CA_\N(q,\Sigma)$.
The authors of~\cite{A-Wijsen24} study rewritings in aggregate logic, which is different from the logic $\Logic_\N$ in the current paper, and obtain a dichotomy similar to our Theorem~\ref{thm:main}. While their rewriting adresses a special case of $\m\CA_\N(q,\Sigma)$, the two approaches differ significantly in their formalism and underlying syntax.
A deeper theoretical exploration is required to precisely pinpoint the differences and commonalities between the two approaches.

Finally, it is also natural to introduce least-upper-bound semantics in the context of semirings.
Specifically, the  \emph{possible answers $\mathrm{M}\CA_\K(\varphi, \Sigma, \DB, \alpha)$ of $\varphi$ on $\DB, \alpha$ with respect to $\Sigma$} is defined as
$\max_{\DB'\in \Rep(\DB,\Sigma)} \varphi(\DB',\alpha)$.
Thus, the possible answers  provide  the tightest upper
bound on the values $\varphi(\DB',\alpha)$ as $\DB'$ ranges over all repairs of $\DB$.

\begin{acks}
We would like to thank the anonymous reviewers for their constructive feedback.
The second and the third author were partially supported by the DFG grant VI 1045-1/1.
\end{acks}

\bibliographystyle{ACM-Reference-Format}
\bibliography{biblio}

\newpage
\appendix

\section{Semiring semantics via interpretations}\label{A:truth}

In this appendix, we give some background on semiring semantics via interpretations 
and provide further  justification for the definition of repairs given in our paper, i.e, why it is reasonable to define repairs using flattening, as done in Section \ref{sec:query}.
The main references are two papers by Gr\"adel and Tannen \cite{gradel2017semiring} and \cite{G-T3} in which semiring semantics to first-order logic \FO is given using the notion of an \emph{interpretation}. We first recall the basic definitions from these two papers.

In what follows,  we assume that ${\mathbb K}=(K, +, \times, 0,1)$ is a positive semiring and  $\tau = (R_1,\ldots,R_m)$ is a relational schema. 

Let $D$ be a finite set and let $\mathsf{Lit}_D$ be the set of all atomic and negated atomic facts involving elements of $D$, i.e., all expressions of the form $R_i({\bf a})$ and $\neg R_i({\bf a})$, where ${\bf a}$ is a tuple of elements from $D$. We will refer to such expressions as \emph{literals} from $D$.
\begin{itemize}
    \item An \emph{interpretation} on $D$ is a function $\pi:\mathsf{Lit}_D \rightarrow K$.
    \item An interpretation $\pi$ is \emph{model defining} if for every atomic fact $R_i({\bf a})$, exactly one of the values $\pi(R_i({\bf a}))$ and $\pi(\neg R_i({\bf a}))$ is different from $0$.
    \item Every model-defining interpretation $\pi$ on $D$ determines a unique finite structure ${\bf D}_\pi=(D,R_1^{{\bf D}_\pi}, \ldots, R_m^{{\bf D}_\pi})$ with universe $D$, where for every $i$ and for every tuple ${\bf a}$ from $D$, we have that
    ${\bf a}\in R_i^{{\bf D}_\pi}$ if and only if $\pi(R_i({\bf a}))\not = 0$.
\end{itemize}
Let $\pi$ be an interpretation on $D$. 
To every \FO-formula $\varphi(x_1,\ldots,x_n)$ in negation normal form (NNF) and every tuple $(a_1,\ldots,a_n)$ from $D$, 
Gr\"adel and Tannen \cite{gradel2017semiring,G-T3}  assign a  value $\pi(\varphi(x_1/a_1,\ldots,x_n/a_n))$ in $K$.  The definition extends the values of the interpretation $\pi$  by induction on the construction of \FO-formulas; in this definition, the addition $+$ operation of $\K$ is used to define the semantics of $\vee$ and $\exists$, while the multiplication $\times$ operation of $\K$ is used to define the semantics of $\wedge$ and $\forall$. In particular, for every \FO-sentence $\psi$ in NNF, the interpretation assigns a semiring value $\pi(\psi)$ to $\psi$. The following proposition (Proposition 5 in \cite{G-T3}) will be useful in the sequel.

\begin{proposition} \label{pro:model-def}
Let $\pi$ be a model-defining interpretation on $D$. Then for every \FO-sentence $\psi$ in NNF, we have that
$$\pi(\psi) \not = 0 \quad \mbox{if and only if} \quad {\bf D}_\pi\models \psi.$$
\end{proposition}

Note that ${\bf D}_\pi$ is an ordinary finite structure. Gr\"adel and Tannen \cite{gradel2017semiring,G-T3} do \emph{not} give semantics of \FO  on ${\mathbb K}$ structures, i.e., on finite structures of the form
${\bf D}=(D,R_1^{\bf A}, \ldots,R_m^{\bf A})$, where each $R_i^{\bf D}$ is a function from $D^{r_i}$ to $K$ and $r_i$ is the arity of the relation symbol $R_i$. In particular, they never define what it means for a $\mathbb K$-structure $\bf D$  or for a $\K$-database $\DB$  to \emph{satisfy} a \FO-sentence $\psi$. Yet, we need such a definition in order to define the notion of a \emph{repair} of a $\mathbb K$-database with respect to a set of integrity constraints. 

\begin{definition} \label{compatible-defn}
Let ${\bf D}=(D,R_1^{\bf D}, \ldots,R_m^{\bf D})$ be a finite $\mathbb K$-structure.
\begin{itemize}
\item We say that an interpretation $\pi$ on $D$ is \emph{compatible with} $\bf D$ if $\pi$ is model-defining and for every $i$ and every tuple ${\bf a}$ from $D$, we have that $\pi(R_i({\bf a}))= R_i^{\bf D}({\bf a})$.
\item The \emph{canonical compatible interpretation with ${\bf D}$} is the interpretation $\pi_{\bf D}$ such that for every $i$ and every tuple ${\bf D}$ from $D$, we have that $\pi(R_i({\bf a}))= R_i^{\bf D}$ and
$$
\pi(\neg R_i({\bf a}))=\begin{cases}
			0, & \text{if $R_i^{\bf D}({\bf a })\not =0$}\\
            1, & \text{if $R_i^{\bf D}({\bf a}) =0$.}
		 \end{cases}
$$
\end{itemize}
\end{definition}

Note that if $\pi$ is an  interpretation that is compatible with a $\mathbb K$-structure $\bf D$, then, by definition, $\pi$ has to agree with each relation $R_i^{\bf D}$ on the atomic facts $R_i({\bf a})$. For the negated atomic facts $\neg R_i({\bf a})$, we have that 
$\pi(\neg R_i({\bf a})) = 0$ if $R_i^{\bf D}({\bf a})\not =  0$, because $\pi$ is model defining.  If, however,  $R_i^{\bf D}({\bf a}) =  0$, then $\pi(\neg R_i({\bf a}))$ can be any non-zero value from $K$; in the case of the canonical compatible interpretation $\pi_{\bf D}$, we have that this non-zero value is $1$.

Let $\bf D$ be a $\mathbb K$-structure and consider the canonical compatible interpretation $\pi_{\bf D}$. As discussed above this gives rise to an ordinary structure ${\bf D}_{\pi_{\bf D}}$ with universe $D$. From the definitions, it follows that
${\bf D}_{\pi_{\bf D}}= \suppo({\bf D})$, where $\suppo({\bf D})=(D, \suppo(R_1^{\bf D}), \ldots, \suppo(R_m^{\bf D}))$ is the ordinary structure with universe $D$ and relations the supports of the relations of $\mathbb D$.

\begin{proposition} \label{sem-prop} Let $\bf D$ be a $\mathbb K$-structure and let $\psi$ be a \FO-sentence. Then the following hold:
\begin{itemize}
\item If  $\pi_1$ and $\pi_2$ are two interpretations on $D$  that are compatible with $\bf D$, then $\pi_1(\psi)\not =  0$ if and only if $\pi_2(\psi) \not = 0$.
\item For every interpretation $\pi$ on $D$ that is compatible with $\bf D$, the following statements are equivalent:
\begin{enumerate}
\item $\pi(\psi) \not = 0$;
\item $\pi_{\bf D}(\psi) \not = 0$;
\item $\suppo({\bf D})\models \psi$.
\end{enumerate}
\end{itemize}
\end{proposition}
\begin{proof}
The first part of the proposition is proved using a straightforward induction and the definition of compatibility. The second part of the proposition follows from the first part of the proposition together with Proposition \ref{pro:model-def}, and the earlier fact that ${\bf D}_{\pi_{\bf D}}= \suppo({\bf D})$.
\end{proof}

The preceding proposition motivates the following definition of what it means for a $\mathbb K$-structure to ``satisfy'' a \FO-sentence.

\begin{definition}\label{sat-defn} We say that a $\bf K$-structure $\bf D$ \emph{satisfies} a \FO-sentence $\psi$ in NNF, denoted ${\bf D}\models_{\mathbb K} \psi$, if 
for some (equivalently, for all) interpretation $\pi$ that is compatible with $\bf D$, we have that $\pi(\psi) \not = 0$. 
\end{definition}

By Proposition \ref{sem-prop}, we have that
${\bf D}\models_{\mathbb K} \psi$ if and only if
$\suppo({\bf A})\models \psi$.

Let us now turn to $\K$-databases. Every $\K$-database $\DB$ with $\K$-relations $R_1^\DB,\ldots,R_m^\DB$ can be viewed as a finite $\K$-structure $\overline{\DB}$ with universe the active domain $D$ of $\DB$ and with the same relations as those of $\DB$, i.e., $\overline{\DB}= (D,R_1^\DB,\ldots,R_m^\DB)$. It is now clear that the semiring semantics of \FO on a $\K$-database $\DB$ that we gave in Section \ref{sec:query} coincides with the definition of the semantics of \FO derived by using the canonical compatible interpretation $\pi_{\overline{\DB}}$.  In particular, Proposition \ref{pro:interp} is a special case of the above Proposition \ref{sem-prop}. Furthermore, in defining what it means for a $\K$-database $\DB$ to satisfy a \FO-sentence $\psi$, we could have used any interpretation 
$\pi$ on $D$ that is compatible with $\DB$, instead of the canonical compatible one. And this amounts to $\suppo(\DB)$ satisfying $\psi$ in the standard sense.

In summary, the preceding considerations  justify using ``flattening'' to define the notion of a repair of a $\mathbb K$-database with respect to a set of key constraints and, more broadly, with respect to a set of \FO-sentences. 

\section{Proofs omitted from the main body}

\subsection{Rewritability of the consistent answers of the path query} \label{appendix_qpath}

Recall that $\qpath = \exists x \exists y \exists z( R(x;y)\land S(y;z))$, for which $\m\CA_\B(\qpath)$ is in $P$~\cite{DBLP:journals/jcss/FuxmanM07}, and has the following
first-order rewriting:
\(
\exists x \exists z' (R(x,z') \land \forall z( R(x,z) \to \exists y (S(z,y)))).
\)
We would like to obtain a similar expression in our setting for \qpath, that is, an expression utilising semiring operations that provides the answer $\m\CA_K(\qpath,\DB)$ for $\K$-databases without having to evaluate \qpath on every  repair of $\DB$. The semiring semantics of \FO on a $\K$-database $\DB$ we gave in Section \ref{sec:query} are precisely the semiring semantics of \FO using the canonical interpretation 
We define the expression $\epath$ as follows:
 $$ \epath \dfn \sum_{a\in D} \min_{b\in D:R^\DB(a,b)\neq 0}({R^\DB(a,b) \times \min_{c\in D:S^\DB(b,c)\neq 0}{S^\DB(b,c)}}),
  $$
 where $\sum$, $\times$, and $\min$ refer to operations of a naturally ordered positive semiring $\K$.
%
%
We claim that
\begin{align}
\m\CA_\K(\qpath,\DB)
= \epath
\label{a:eq:rewriting_qpath}
\end{align}
for every $\K$-database $\DB$. 
%
%
%

\begin{proof}
We show that the equality holds.
Note that 
    \[
    \m\CA_\K(\qpath,\DB) =
    \min_{\DB' \in \Rep(\DB)} \FOKeval{\qpath}{\DB'} =
    \min_{\DB' \in \Rep(\DB)} \sum_{a,b,c\in D'} R^{\DB'}(a,b) \times S^{\DB'}(b,c),
    \]
where the former equality is the definition of $\m\CA_\K$ and the latter follows from the semiring semantics for CQs.  
%
To show (\ref{a:eq:rewriting_qpath}), 
it suffices to establish that the following two statements hold:
\begin{enumerate}[label=(\roman*)]
    \item For every repair $\DB'$ of $\DB$, we have that $\epath \leq_\K \qpath(\DB')$; \label{a:qpath:first}
    \item There is a repair $\DB^*$ of $\DB$ such that $\epath = q(\DB^*)$. \label{a:qpath:second}
\end{enumerate}

 For \ref{a:qpath:first}, let $\DB'\in\Rep(\DB)$ and $a,b,c\in D'$ be arbitrary,
 and consider the expression
 $R^{\DB'}(a,b) \times S^{\DB'}(b,c)$.
 We distinguish the following three cases: 

 Case 1: $a$ is not a key value for $R^\DB$.

 Case 2: $a$ is a key value for $R^\DB$, but $b$ is not a key value for $S^\DB$.

 Case 3. $a$ is a key value for $R^\DB$ and $b$ is a key value for $S^\DB$.

 Note that, in the sum defining \epath,
 the element $a$ contributes only one summand to that sum. In Case 1, in which $a$ is not a key value for $R^\DB$,  the summand to which $a$ contributes has value~$0$, since we set $\min(\emptyset) = 0$. 
 Similarly, in Case 2, in which $a$ is a key value for $R^\DB$ but $b$ is not a key value for $S^\DB$, the factor 
 $\min_{z\in D:S^\DB(b,z)\neq 0}{S^\DB(b,z)}$ to which $b$ contributes takes value $0$, and hence the summand to which $a$ contributes takes value $0$. In Case 3, in which $a$ is a key value for $R^\DB$
and $b$ is a key value for $S^\DB$, by monotonicity of multiplication (if $j'\leq_\K j$ then $i\times j' \leq_\K i\times j$) we have
\[
\min_{y\in D:R^\DB(a,y)\neq 0}({R^\DB(a,y) \times \min_{z\in D:S^\DB(y,z)\neq 0}{S^\DB(y,z)}}) \leq_\K 
R^{\DB'}(a,b)\times S^{\DB'}(b,c).
\]
Therefore, the  summand of \epath
to which $a$ contributes is dominated by the summand of $\qpath(\DB')$ to which $a$ contributes. 
Hence $\epath \leq_\K \qpath(\DB')$, which concludes the proof of \ref{a:qpath:first}.

For \ref{a:qpath:second}, we distinguish
   two cases:  $\DB\not \models_\K\qpath$ or 
 $\DB\models_\K\qpath$.
 If $\DB\not\models_\K\qpath$, then
 for every repair 
 $\DB'$ of $\DB$, we have that
 $\DB'\not \models_\K\qpath$,  Hence, $\qpath(\DB') = 0$, for every repair $\DB'$ of $\DB$.
  Now, an inspection for $\epath$ 
  shows that $\epath = 0$. 
  Therefore, we can pick any repair $\DB^*$ of $\DB$ (at least one repair  exists) and conclude that
  $\epath = 0 = \qpath(\DB^*)$.

 Finally, suppose that 
  $\DB \models_\K\qpath$.
Our goal is to show that there exists a repair $\DB^*$ of $\DB$ such that 
$\epath = q(\DB^*)$.
    We build a repair $\DB^*$ as follows. For every element $b$ that is a key value for $S^\DB$, we choose a value $c^*$ such that $S^\DB(b,c^*) = \min\{ S^\DB(b,c) \mid S^\DB(b,c) \neq 0, c\in D \}$.
    Now, for every element $a$ that is a key value for $R^\DB$, we have two possibilities:
    \begin{enumerate}[label=(\roman*)]
        \item There exists a $b$ such that $R^\DB(a,b) \neq 0$, and $S^\DB(b,c) = 0$ for every $c\in D$ (i.e.,  
        $b$ is not a key value for $S^\DB$). In this case,  we pick one such $b$ and put $(a,b)$ in $R^{\DB^*}$. \label{item:1_qpath}
        \item For every $b$ such that $R^\DB(a,b) \neq 0$, there exists a value $c$ for which $S^\DB(b,c) \neq 0$. In this case, we choose an element  $b^*$ with $R^\DB(a,b^*) \neq 0$ and such that the value $R^\DB(a,b^*) \times S^\DB(b^*,c^*)$ is minimised, and we put $(a, b^*)$ in $R^{\DB^*}$. \label{item:2_qpath}
    \end{enumerate}
    
    Let $a$ be a key value for $R^\DB$ for which the first possibility holds. If $b$ is the element for which $R^{\DB^*}(a,b)$ is a fact of $\DB^*$, then   $\min\{S^\DB(b,c): S^\DB(b,c)\neq 0, c\in D\} = 0$.
     Consequently,  we have that 
    $$\min_{y\in D:R^\DB(a,y)\neq 0}({R^\DB(a,y) \times \min_{z\in D:S^\DB(y,z)\neq 0}{S^\DB(y,z)}})= 0.$$
    Moreover, $R^{\DB^*}(a,b) \times S^{\DB^*}(b,c) = 0$, for any $c\in D^*$. 
    Thus, the value of the summand in $\epath$ 
    contributed by $a$ equals the value of the summand in $\qpath(\DB^*)$ contributed by $a$. 
    
Let  $a$ be a key value such that the second possibility holds.  From the definition of $\DB^*$, it follows  that the expression $\min_{y\in D:R^\DB(a,y)\neq 0}({R^\DB(a,y) \times \min_{z\in D:S^\DB(y,z)\neq 0}{S^\DB(y,z)})}$ in $\epath$ 
is equal to the value in $\qpath(\DB^*)$ contributed by  $a$. 

By combining the findings in these two possibilities, we conclude that 
$\epath = \qpath(\DB^*)$.
\end{proof}

\subsection{$\K$-circuits and complexity theory}\label{appendix_circuits}


\circuitprop*


\begin{proof}
The proof proceeds analogously to the classical case proving that $\FO$ is in $\DLOGTIME$-uniform $\aco$. When considered as inputs to circuits, $\K$-databases are encoded as strings of semiring values in the same fashion as Boolean databases are encoded as strings of Booleans (see, e.g., \cite{Libkin04}).
The inputs that we consider also include an assignment $\alpha$ giving values for some fixed finite set $X$ of variables. 
If $R^\DB$ is a $\K$-relation of arity $k$ in $\DB$, then its encoding $\enc(R^\DB,\alpha)$ is simply the concatenation of the semiring values $R^\DB(\alpha(\vec{a}))$, for $\vec{a}\in (D\cup X)^k$, written in some predefined order (that is we stipulate some ordering on $D\cup X$ and use that to define an ordering of $(D\cup X)^k$).
The encoding $\enc(\DB,\alpha)$ of the $\K$-database and an assignment is then the concatenation of the encodings of its $\K$-relations $\enc(R^{\DB},\alpha)$ in some predefined order.

Given $n\in\N$ and a formula $\varphi \in\Logic_\K$, one can recursively define the $\aco_\K(+,\times_2, \min)$-circuit that computes the value of $\varphi$ in a $\DB,\alpha$, such that $\lvert D\rvert = n$, on the input $\enc(\DB, \alpha)$. In the transformation of $\varphi \in \Logic_\K$ to a circuit, it suffices to transform subformulas of the forms $\exists x. \varphi(\vec{y},x)$ and $\nabla x. \varphi(\vec{y},x)$ to expressions $\sum_{a\in D} \varphi(\vec{y},a/x)$ and $\min_{a\in D} \varphi(\vec{y},a/x)$, respectively. After this, the remaining construction of the circuit is to treat each subformula as a gate of the circuit labelled with its top-most connective. Gates corresponding to atomic formulas are input gates and are labelled with an appropriate part of the input $\enc(\DB,\alpha)$ determined by the ordering used in it. The argument that there is a $\DLOGTIME$ algorithm that describes $C_n$, given $n$, is the same as in the classical case for $\FO$ (see \cite{DBLP:conf/wollic/BarlagV21} for a similar proof for $\FO_\R$ and $\DLOGTIME$-uniform $\aco_\R(+,\times)$).
\end{proof}

\subsection{Acyclicity of the attack graph and semiring rewriting of the consistent answers}\label{appendix_rewriting}

\jefslemmatwo*

\begin{proof}
Notice that the attack graphs of $q[x]$ and $q$ have the same atoms, hence they have the same nodes. Consider an edge in the attack graph of $q[x]$, that is, an edge between the nodes corresponding to atoms $R(\vec{y};\vec{z})$ and $R'(\vec{y'};\vec{z'})$. Thus, $R(\vec{y}; \vec{z})$ attacks a variable $y'_i$ in $\var(\vec{y'})$ that is bounded in $q[x]$. In particular, $y'_i$ is bounded in $q$, and both atoms $R(\vec{y}; \vec{z})$ and $R'(\vec{y'};\vec{z'})$ are in $q$. Since $R(\vec{y}; \vec{z})$ attacks $y'_i$, there are witnesses $x_1, \ldots, x_m=y'_i$ such that:
\begin{enumerate}
     \item $x_1 \in \var(\vec{z})$ (i.e., is a non-key variable of $R$), and $x_m=y'_i$;
     \item for all $i <n$ we have that $x_i, x_{i+1}$ occur together in some atom of $q[x]$ (thus, they occur together in some atom of $q$); and
     \item For every $i\leq n$ we have that $x_i \not\in \{ \var(\vec{y})\}^+_{\Sigma(q[x]\setminus R)}$. Notice that the set of key-constraints in $q[x] \setminus \{R(\vec{y}; \vec{z})\}$ and the set of key constraints in $q\setminus \{R(\vec{y}; \vec{z})\}$ coincide .
 \end{enumerate}
We can conclude that $R(\vec{y}; \vec{z})$ attacks $y'_i$ in the attack graph of $q$ as well, in other words, the edge between $R(\vec{y};\vec{z})$ and $R'(\vec{y'};\vec{z'})$ is also in the attack graph of $q$, therefore the attack graph of $q[x]$ is a subgraph of the attack graph of $q$.
\end{proof}


\lemmaKDBs*

\begin{proof}
    Note that, by definition
    \[
     \m\CA_\K(q, \DB, \alpha) = \min_{\DB' \in \Rep(\DB,\Sigma)} q(\DB', \alpha) = \min_{\DB' \in \Rep(\DB,\Sigma)} \sum_{c \in D'} q[x](\DB', \alpha(c/x)).
    \]
    Hence, it suffices to establish that
    \[
    \min_{\DB' \in \Rep(\DB,\Sigma)} \sum_{c \in D'} q[x](\DB', \alpha(c/x)) = \sum_{c \in D} \min_{\DB' \in \Rep(\DB,\Sigma)} q[x](\DB', \alpha(c/x)).
    \]
    It is easy to see that
    \[
    \min_{\DB' \in \Rep(\DB,\Sigma)} \sum_{c \in D'} q[x](\DB', \alpha(c/x)) \geq_\K \sum_{c \in D} \min_{\DB' \in \Rep(\DB,\Sigma)} q[x](\DB', \alpha(c/x)).
    \]


    For the other direction, if $\m\CA_\K(q, \DB,\alpha) = 0$, then the lemma holds. Suppose $\m\CA_\K(q, \DB, \alpha) >0$. 
    We need to show that
    \begin{equation}\label{eq:0}
    \min_{\DB' \in \Rep(\DB,\Sigma)} \sum_{c \in D'} q[x](\DB', \alpha(c/x))
    \leq_\K \sum_{c \in D} \min_{\DB' \in \Rep(\DB, \Sigma)} q[x](\DB', \alpha(c/x)).
    \end{equation}
    For every repair $\RP$ of $\DB$, set
    \(
    \rho(\RP) \dfn \{ c \in D \mid \RP,\alpha(c/x) \models_{\K} q[x] \} 
    \)
    and $v(\RP) \dfn \sum_{c \in R} q[x](\RP, \alpha(c/x))$. Since $\m\CA_\K(q, \DB, \alpha) >0$, which is the minimum $v(\RP)$ over the repairs $\RP$ of $\DB$, we have that $v(\RP) >0$ for every $\RP$.
    Let $m= \min_{\RP \in \Rep(\DB,\Sigma)} v(\RP) = \m\CA_\K(q, \DB, \alpha)$ and fix $\RP_{\min} \in \Rep(\DB, \Sigma)$ such that $v(\RP_{\min}) = m$.
        We will show the following claim:


    \begin{claim}\label{claim:lemmaKdbs}
    If $a \in \rho(\RP_{\min})$ and $\RP\in \Rep(\DB,\Sigma)$, then
    $q[x](\RP_{\min},\alpha(a/x)) \leq_\K q[x](\RP,\alpha(a/x))$.  
    \end{claim} 
    \begin{proof}[Proof of Claim~\ref{claim:lemmaKdbs}]
     Let $(R_{1}, R_{2}, \ldots, R_{n})$ be a topological ordering of the attack graph of~$q$.
     For simplicity, when we use a relation name $R_{i}$ in contexts where an atom is expected, we mean the unique $R_{i}$-atom of $q$. 
    If $\RP$ is a repair, and $i\in\{0,1,2,\ldots,n\}$, then we write $\agree{\RP}{i}$ for the smallest subset of $\Rep(\DB,\Sigma)$ that contains $\RP'$ whenever $\RP$ and $\RP'$ have exactly the same $R_{j}$-facts for all $j\in\{1, 2, \ldots, i\}$.
    Thus, $\agree{\RP}{0}=\Rep(\DB,\Sigma)$, and  $\agree{\RP}{n}=\{\RP\}$.
We introduce some convenient terminology.
 An \emph{$R_{\ell}$-block} is a maximal set of $R_{\ell}$-facts of $\suppo(\DB)$ that agree on all key attributes.
    Clearly, for each $\ell\in\{1,2,\ldots,n\}$, every repair selects exactly one fact from each $R_{\ell}$-block.
     An \emph{embedding of $q$ into $\RP$} is an assignment that maps every atom of $q$ to a fact in $\suppo(\RP)$. 
      If $\theta$ is an embedding, then $\theta(R_{\ell})$ denotes the fact to which the atom~$R_{\ell}$ is mapped (where $\ell\in\{1,2,\ldots,n\}$).
    A repair $\RP$ is called \emph{superfrugal}~\cite{A-Wijsen24} if no repair $\RP'$ has a set of embeddings that is a strict subset of that of $\RP$.
    An embedding in a superfrugal repair is also called a \emph{$\forall$embedding}. 
    If $\beta$ is a $\forall$embedding, then,  for all $i\in\{1,2,\ldots,n\}$, the $R_{i}$-block that contains $\beta(R_{i})$ is called a \emph{$\forall R_{i}$-block}. 
    From~\cite[Lemma~4.5]{A-Wijsen24}, it follows that for all $i\in\{1,2,\ldots,n\}$, we can select a fact from each block that is not a $\forall R_{i}$-block such that the selected facts do not belong to any embedding into the resulting repair. Therefore, in what follows, it suffices to consider only $\forall R_{i}$-blocks. 

    A repair $\RP$ such that $v(\RP)=m$ can be constructed as follows, for decreasing values of $i=n, n-1, \ldots, 1$:
\begin{itemize}
    \item
    for $i=n$, select in each $\forall R_{n}$-block the fact with the smallest annotated semiring value; and
    \item
    for $i<n$, in each $\forall R_{i}$-block, select the fact that yields the smallest annotated semiring value, given the already fixed $R_{j}$-facts for $j\in\{i+1, i+2, \ldots, n\}$.
    This selection is illustrated next and, as we will argue shortly, is independent of the $R_{\ell}$-facts for $\ell<i$.
\end{itemize}    

\begin{example}\label{ex:independentchoice}
Consider the bag semiring.
Assume that the last two facts of $q$, in the topological order of the attack graph, are $R_{n-1}(v;y)$ and $R_{n}(y,z)$, where all attributes of $R_{n}$ are key attributes.
Assume that the only $R_{n}$-facts with positive support are $R_{n}^{\DB}(b,d_{1})=1$,  $R_{n}^{\DB}(b,d_{2})=1$, and  $R_{n}^{\DB}(c,d_{1})=8$.
Furthermore, assume that $R_{n-1}^{\DB}(a,b)=2$ and $R_{n-1}^{\DB}(a,c)=1$, two facts belonging to the same $\forall R_{n-1}$-block. 
The former $R_{n-1}$-fact contributes $(2\times 1)+(2\times 1)=4$, while the latter contributes $1\times 8=8$. So our procedure will select $R_{n-1}(a,b)$, which, notably, is not the fact with the smallest annotated semiring value in its block.
\end{example}    

It can now be argued by induction on decreasing $i=n, n-1,\ldots,1$ that every superfrugal repair~$\RP$ containing the selected $R_{j}$-facts for every $j\in\{i, i+1, \ldots, n\}$ minimizes the semiring value of $q$ across all repairs in $\agree{\RP}{i-1}$. 
The reasoning is similar to that in Claim~\ref{claim:rho_containment}.
Assume that in some superfrugal repair $\RP$, we replace a fact $B$ with $C$ in some $\forall R_{i}$-block $\block$, such that the semiring value of $q$ decreases.
Let~$\RP_{C}$ be the repair satisfying $\suppo(\RP_{C})=\lrformula{\suppo(\RP)\setminus\{B\}}\cup\{C\}$.
There will be an embedding $\beta$ into $\RP$ such that $\beta(R_{i})=B$, as well as an embedding $\gamma$ into $\RP_{C}$ such that $\gamma(R_{i})=C$.
    Since $B$ and $C$ agree on their key, it follows that $\beta$ and~$\gamma$ agree on every variable in $V\eqdef\{v\in\var(q)\mid\Sigma(q\setminus\{R_i\})\models\key{R_{i}}\to v\}$. 
    Define $\compose{\gamma}{\beta}$ as the assignment such that for every variable~$v$ in $\var(q)$,
    \[
    \compose{\beta}{\gamma}(v) = 
    \begin{cases}
        \gamma(v) & \textnormal{if $R_{i}$ attacks $v$};\\
        \beta(v) & \textnormal{otherwise}.
    \end{cases}
    \]
    Since $x$ is unattacked in $q$, we have $\compose{\beta}{\gamma}(x)=\beta(x)$.
    Notably, $\compose{\beta}{\gamma}$ is an embedding into $\RP_{C}$ because if an atom~$R_{\ell}$ of $q$ (where $\ell\in\{1,2,\ldots,n\}$) contains both variables attacked and unattacked by~$R_{i}$, then~$\beta$ and $\gamma$ agree on the unattacked variables, as they belong to~$V$; consequently, $\compose{\beta}{\gamma}(R_{\ell})=\gamma(R_{\ell})$. 
    Furthermore, if $\ell\neq i$ and $R_{\ell}$ is not attacked by $R_{i}$, then $\compose{\beta}{\gamma}(R_{\ell})=\beta(R_{\ell})$, which, in particular, holds for $\ell\in\{1,2,\ldots,i-1\}$.
    Note also that $x$ occurs in at least one atom that is not attacked by~$R_{i}$.
    Informally, we conclude that if $B$ is replaced by~$C$, then \emph{every} embedding into $\RP$ that used $B$ can be transformed into one that uses $C$ instead, while remaining unchanged over all atoms not attacked by $R_{i}$.
    Consequently, the optimal (minimizing) choice from $\block$ is independent of atoms not attacked by $R_{i}$.
    In particular, regarding $x$, if $\beta_{1}$ and~$\beta_{2}$ are distinct embeddings into $\RP$ such that $\beta_{1}(R_{i})=\beta_{2}(R_{i})=B$, then the optimal choice in $\block$ is the same for both $\beta_{1}$ and $\beta_{2}$, even if $\beta_{1}(x)\neq\beta_{2}(x)$.  
    For example, the selection from the $R_{n-1}$-block in Example~\ref{ex:independentchoice} does not require knowledge of the $R_{\ell}$-facts for $\ell<n-1$. 


When $i$ reaches $1$, the superfrugal repair with the selected facts minimizes the semiring value of~$q$ across all repairs of $\agree{\RP}{0}=\Rep(\DB,\Sigma)$. 
Moreover, our minimization procedure ensures that for every  $a \in D$, the semiring value of $q[x]$ with respect to $\alpha(a/x)$ is also minimized. 
This concludes the proof of Claim~\ref{claim:lemmaKdbs}.
     \end{proof}

Equation~\eqref{eq:0} now follows from Claim~\ref{claim:lemmaKdbs}, and this finishes the proof of the lemma.
\end{proof}

\end{document}